\documentclass[a4paper]{jpconf}
\usepackage{amssymb}
\usepackage{graphicx}
\usepackage{amsmath}
\usepackage{makeidx}
\usepackage{indentfirst}

\newcounter{resultnum}[section]
\newtheorem{conclusion}{Conclusion}[section]

\newcounter{conclusionnum}[section]
\newcounter{conditionnum}[section]
\newcounter{conjecturenum}[section]
\newcounter{examplenum}[section]
\newcounter{exercisenum}[section]
\newtheorem{lemma}{Lemma}[section]

\newcounter{lemmanum}[section]
\newcounter{notationnum}[section]
\newtheorem{theorem}{Theorem}[section]

\newcounter{theoremnum}[section]
\newtheorem{definition}{Definition}[section]

\newcounter{definitionnum}[section]
\newtheorem{corollary}{Corollary}[section]

\newcounter{corollarynum}[section]
\newtheorem{remark}{Remark}[section]

\newcounter{remarknum}[section]
\newtheorem{proposition}{Proposition}[section]

\newcounter{propositionnum}[section]
\newcounter{acknowledgementnum}[section]
\newcounter{algorithmnum}[section]
\newcounter{axiomnum}[section]
\newtheorem{claim}{Claim}[section]

\newcounter{claimnum}[section]
\newcounter{summarynum}[section]
\newcounter{problemnum}[section]
\newenvironment{proof}[1][]{\textbf{Proof.} }{}

\begin{document}

\title{Decoupling of Field Equations in
Einstein and Modified Gravity}

\author{ {Sergiu I. Vacaru} }

\address{Alexandru Ioan Cuza University at Ia\c si (Yassy), UAIC,\\ Alexandru Lapu\c sneanu street, nr. 14, Corpus R, office 323;
Ia\c si, Romania, 700057}

\ead{sergiu.vacaru@uaic.ro}


\begin{abstract}
This paper is concerned with giving the proof that there is a general
decoupling property of vacuum and nonvacuum gravitational field equations in
Einstein gravity and $f(R,T)$--modifications. The constructions are possible
in terms of geometric and physical objects adapted to certain nonholonomic
2+2 splitting with local fibred structure. This allows us to generate exact
and/or parametric depending solutions with generic off--diagonal metrics and
generalized, or Levi--Civita, connections. Different classes of modified
spacetimes are determined by corresponding generating and integration
functions depending, in general, on all space and time coordinates and may
possess, or not, Killing symmetries. The initial data sets for the Cauchy
problem and their global properties are analyzed. There are formulated the
criteria of evolution with spacetime splitting and decoupling of fundamental
field equations. Examples of exact solutions defining ellipsoid deformations
of black hole metrics and solitonic configurations are provided.
%



\end{abstract}


\section{Introduction}

The issue of constructing exact and approximate solutions of (modified)
gravitational and matter field equations is of interest in
mathematical physics and for various applications in modern cosmology and
astrophysics. Mathematically, the fundamental field equations in gravitational
physics are defined by sophisticate systems of nonlinear partial
differential equations (PDE) which are very difficult to be integrated and
studied in general forms.

In this work, we address again and develop a geometric approach (the
so--called anholonomic frame deformation method, AFDM, see \cite%
{ijgmmp,vexsol1,vexsol2}) to constructing exact solutions in gravity
described by generic off--diagonal and nonlinear gravitational, gauge and
scalar field interactions. Such a geometric method allows us to generate
exact solutions when the coefficients of generic off--diagonal metrics and
various classes of connections depend on all spacetime coordinates. Our goal
is to provide new results on the decoupling property\footnote{%
It is used also the term \textquotedblright separation\textquotedblright\ of
equations, which should not be confused with separation of variables.} of
modified gravity equations and formal integration of such nonlinear systems
of PDEs. Certain issues on the Cauchy problem and decoupling of
gravitational equations and the initial data sets and nonholonomic evolution
(with non--integrable constraints and/or with respect to anholonomic frames)
will be analyzed. Examples of off--diagonal solutions for modified black
holes/ellpsoids and solitonic configurations will be given.

Section \ref{s2} contains geometric preliminaries on modified gravity
theories written in nonholonomic variables. In Section \ref{s3} we provide
the main Theorems on decoupling the gravitational field equations for
generic off--diagonal metrics with one Killing symmetry. We consider
extensions to "non--Killing" configurations with coefficients depending on
all set of four coordinates in section \ref{s4}. Section \ref{s5} contains a
study of the Cauchy problem in connection to the decoupling property of
gravitational field equations. Explicit examples of generic off--diagonal
exact solutions are studied in the next two sections. In Section \ref{s6}, there
are constructed nonholonomic vacuum deformations of black holes determined
by off--diagonal gravitational interactions. Ellipsoid--solitonic
configurations are analyzed in Section \ref{s7}. The most important formulas
and computations which are necessary to prove the decoupling property are
given in Appendix \ref{sb}.

\section{Einstein and Modified Gravity in Nonholonomic Variables}

\label{s2} We provide geometric preliminaries on $f(\ ^{s}R,T)$--gravity
models written in nonholonomic variables.\footnote{%
Such a formulation is necessary for proving the main results on decoupling
and integrating in certain general forms the gravitational field equations,
see sections \ref{s3} and \ref{s4}.} In this work, the scalar curvature $\
^{s}R$ is constructed for an auxiliary connection $\mathbf{D}$ completely
defined by the metric structure $\mathbf{g}$ and $T$ is the trace of the
energy--momentum tensor for matter fields. If $\mathbf{D}=\nabla ,$ where $%
\nabla $ is the Levi--Civita connection, we get as particular cases the
theories studied, for instance, in Ref. \cite{odintsov}, when $\ ^{s}R=R$ is
just the scalar curvature of a pseudo--Riemannian space. For $f(\
^{s}R,T)=R, $ such a modified gravity model transforms into the "standard"
Einstein gravity theory.

\subsection{Nonholonomic 2+2 splitting}

In general relativity (GR), the curved spacetime $\left( V,\mathbf{g}\right)
$ is defined by a pseudo--Riemannian manifold $V$ endowed with a Lorentzian
metric $\mathbf{g}$ as a solution of the Einstein equations\footnote{%
We assume that readers are familiar with basic concepts and results on
mathematical relativity and methods of constructing exact solutions
outlined, for instance, in above mentioned monographs and reviews.}. The
space of classical physical events is modelled as a Lorentzian four
dimensional, 4-d, manifold $\mathcal{V}$ (of necessary smooth class,
Housdorf and paracompact one) when the symmetric 2--covariant tensor $%
\mathbf{g}=\{g_{\alpha \beta }\}$ defines in each point $u\in \mathcal{V}$
and nondegenerate bilinear form on the tangent space $T_{u}\mathcal{V},$ for
instance, of signature $(+++-)$. The assumption that $T_{u}\mathcal{V}$ has
its prototype (local fiber) the Minkowski space $\mathbb{R}^{3,1}$ leads at
a causal character of positive/negative/null vectors on $\mathcal{V},$ i.e.
for the module $\mathcal{X}$ of vectors fields $X,Y,...\in \mathcal{X}(%
\mathcal{V}),$ which is similar to that in special relativity.

Let us denote by $e_{\alpha }$ and $e^{\beta }$ a local frame and,
respectively, its dual frame [we can consider orthonormal (co) bases], where
Greek indices $\alpha ,\beta ,..$ may be abstract ones, or running values $%
1,2,3,4.$ For a coordinate base $u=\{u^{\alpha }\}$ on a chart $U\subset
\mathcal{V},$ we can write $e_{\alpha }=\partial _{\alpha }=\partial
/\partial u^{\beta }$ and $e^{\beta }=du^{\beta }$ and, for instance, define
the coefficients of a vector $X$ and a metric $\mathbf{g,}$ respectively, in
the forms $X=X^{\alpha }e_{\alpha }$ and
\begin{equation}
\mathbf{g}=g_{\alpha \beta }(u)e^{\alpha }\otimes e^{\beta },  \label{mst}
\end{equation}%
where $g_{\alpha \beta }:=\mathbf{g}(e_{\alpha },e_{\beta }).$\footnote{%
The summation rule on repeating low--up indices will be applied if the
contrary is not stated.} We consider bases with non--integrable
(equivalently, nonholonomic/anholonomic) $2+2$ splitting for conventional,
horizontal, h, and vertical, v, decomposition, when for the tangent bundle $T%
\mathcal{V}$ $:=\bigcup\nolimits_{u}T_{u}\mathcal{V}$ a Whitney sum
\begin{equation}
\mathbf{N}:\ T\mathcal{V}=h\mathcal{V}\oplus v\mathcal{V}  \label{ncon}
\end{equation}%
is globally defined. Such a nonholonomic distribution is determined locally
by its coefficients $N_{i}^{a}(u),$ when $\mathbf{N=}N_{i}^{a}(x,y)dx^{i}%
\otimes \partial /\partial y^{a},$ where $u^{\alpha }=(x^{i},y^{a})$ splits
into h--coordinates, $x=(x^{i}),$ and v--coordinates, $y=(y^{a}),$ with
indices running, respectively, values $i,j,k,...=1,2$ and $a,b,c,...=3,4.$%
\footnote{We note that the 2+2 splitting can be considered as an alternative to the
well known 3+1 splitting. The first one is convenient, for instance, for
constructing generic off--diagonal solutions and elaborating models of
deformation and/or A--brane quantization of gravity, but the second one is
more important for canonical/loop quantization etc.}

A spacetime $\left( V,\mathbf{g}\right) $ can be equipped with a
non--integrable fibred structure (\ref{ncon}) and such a manifold is called
nonholonomic (equivalently, N--anholo\-nomic). We use ''boldface'' letters
in order to emphasize that certain spaces and geometric
objects/constructions are ''N--adapted'', i.e. adapted to a
h--v--splitting. The geometric objects are called distinguished (in brief,
d--objects, d--vectors, d--tensors etc). For instance, we write a d--vector
as $\mathbf{X}=(hX,vX)$ for a nonholonomic Lorentz manifold/spacetime $%
\left( \mathbf{V},\mathbf{g}\right) $.

On a spacetime $\left( \mathbf{V},\mathbf{g}\right) ,$ we can perform/adapt
the geometric constructions using "N--elongated"' local bases (partial
derivatives), $\mathbf{e}_{\nu }=(\mathbf{e}_{i},e_{a}),$ and cobases
(differentials), $\mathbf{e}^{\mu }=(e^{i},\mathbf{e}^{a}),$ when
\begin{eqnarray}
\mathbf{e}_{i} &=&\partial /\partial x^{i}-\ N_{i}^{a}(u)\partial /\partial
y^{a},\ e_{a}=\partial _{a}=\partial /\partial y^{a},  \label{nader} \\
\mbox{ and \  }e^{i} &=&dx^{i},\ \mathbf{e}^{a}=dy^{a}+\ N_{i}^{a}(u)dx^{i}.
\label{nadif}
\end{eqnarray}%
Such (co) frame structures depend linearly on N--connection coefficients
being, in general, nonholonomic. For instance, the basic vectors (\ref{nader}%
) satisfy certain nontrivial nonholonomy relations
\begin{equation}
\lbrack \mathbf{e}_{\alpha },\mathbf{e}_{\beta }]=\mathbf{e}_{\alpha }%
\mathbf{e}_{\beta }-\mathbf{e}_{\beta }\mathbf{e}_{\alpha }=W_{\alpha \beta
}^{\gamma }\mathbf{e}_{\gamma },  \label{nonholr}
\end{equation}%
with  nontrivial anholonomy coefficients
 $W_{ia}^{b}=\partial _{a}N_{i}^{b},W_{ji}^{a}=\Omega _{ij}^{a}=\mathbf{e}%
_{j}\left( N_{i}^{a}\right) -\mathbf{e}_{i}(N_{j}^{a})$.

Any spacetime metric $\mathbf{g}=\{g_{\alpha \beta }\}$ (\ref{mst}), via
frame/coordinate transforms can be represented equivalently in N--adapted
form as a d--metric
\begin{equation}
\ \mathbf{g}=\ g_{ij}(x,y)\ e^{i}\otimes e^{j}+\ g_{ab}(x,y)\ \mathbf{e}%
^{a}\otimes \mathbf{e}^{b},  \label{dm}
\end{equation}%
or, with respect to a coordinate local cobasis $du^{\alpha
}=(dx^{i},dy^{a}), $ as an off--diagonal metric
 $\mathbf{g}=\underline{g}_{\alpha \beta }du^{\alpha }\otimes du^{\beta }$,
where
\begin{equation}
\underline{g}_{\alpha \beta }=\left[
\begin{array}{cc}
g_{ij}+N_{i}^{a}N_{j}^{b}g_{ab} & N_{j}^{e}g_{ae} \\
N_{i}^{e}g_{be} & g_{ab}%
\end{array}%
\right] .  \label{ansatz}
\end{equation}%
A metric $\mathbf{g}$ is generically off--diagonal, if (\ref{ansatz}) can not be
diagonalized via coordinate transforms. Ansatzes of this type are used in
Kaluza--Klein gravity when $N_{i}^{a}(x,y)=\Gamma _{bi}^{a}(x)y^{a}$ and $%
y^{a}$ are "'compactified" extra--dimensions coordinates, or in Finsler
gravity theories, see details in \cite{overduin,vexsol2}. In this work, we
restrict our considerations only to the 4--d gravity theories. The principle
of general covariance in GR, allows us to consider any frame/coordinate
transforms and write a spacetime metric $\mathbf{g}$, equivalently, in any of the
above form (\ref{mst}), (\ref{ansatz}) and/or (\ref{dm}). The last mentioned
parametrization will allow us to prove a very important property of
decoupling of gravitational field equations with respect to N--adapted bases
(\ref{nader}) and (\ref{nadif}).

Via frame/coordinate transforms $e_{\alpha }=e_{\ \alpha }^{\alpha ^{\prime
}}(x,y)e_{\alpha ^{\prime }},$ $\underline{g}_{\alpha \beta }=e_{\ \alpha
}^{\alpha ^{\prime }}e_{\ \beta }^{\beta ^{\prime }}\ \underline{g}_{\alpha
^{\prime }\beta ^{\prime }},$ a metric $\mathbf{g}$ (\ref{dm}) can be
written in a form with separation of v--coordinates and nontrivial vertical
conformal transforms,
\begin{eqnarray}
\mathbf{g} &=&g_{i}dx^{i}\otimes dx^{i}+\omega ^{2}h_{a}\underline{h}_{a}%
\mathbf{e}^{a}\otimes \mathbf{e}^{a},  \label{ans1} \\
\mathbf{e}^{3} &=&dy^{3}+\left( w_{i}+\underline{w}_{i}\right) dx^{i},\
\mathbf{e}^{4}=dy^{4}+\left( n_{i}+\underline{n}_{i}\right) dx^{i},  \notag
\end{eqnarray}%
\begin{eqnarray}
\mbox{were\ }g_{i} &=&g_{i}(x^{k}),g_{a}=\omega ^{2}(x^{i},y^{c})\
h_{a}(x^{k},y^{3})\underline{h}_{a}(x^{k},y^{4}),  \notag \\
N_{i}^{3} &=&w_{i}(x^{k},y^{3})+\underline{w}%
_{i}(x^{k},y^{4}),N_{i}^{4}=n_{i}(x^{k},y^{3})+\underline{n}%
_{i}(x^{k},y^{4}),  \label{paramdcoef}
\end{eqnarray}%
are functions of necessary smooth class which will be defined in a form to
generate solutions of gravitational field equations\footnote{%
There is not summation on repeating "low" indices $a$ in formulas (\ref{paramdcoef}) but such a
summation is considered for crossing \textquotedblright
up--low\textquotedblright\ indices $i$ and $a$ in (\ref{ans1})). We shall
underline a function if it positively depends on $y^{4}$ but not on $y^{3}$
and write, for instance, $\underline{n}_{i}(x^{k},y^{4})$.}.

The aim of this section is to prove that the gravitational field equations
in the vacuum cases and certain very general classes of matter field sources
decouple for parameterizations of metrics in the form (\ref{paramdcoef}).

\subsection{Torsions and curvatures}

In a general case, a metric--affine manifold $V$ is endowed with a metric
structure $\mathbf{g}$ and an affine (linear) connection structure $D$ (as a
covariant differentiation operator). A linear connection gives us the
possibility to compute the directional derivative $D_{X}Y$ of a vector field
$Y$ in the direction of $X$. It is characterized by three fundamental
geometric objects

\begin{enumerate}
\item the torsion field is (by definition) $\mathcal{T}(X,Y):=D_{\mathbf{X}%
}Y-D_{\mathbf{Y}}X-[X,Y];$

\item the curvature field is $\mathcal{R}(X,Y):=D_{\mathbf{X}}D_{\mathbf{Y}%
}-D_{\mathbf{Y}}D_{\mathbf{X}}-D_{\mathbf{[X,Y]}};$

\item the nonmetricity field is $\mathcal{Q}(X):=D_{\mathbf{X}}\mathbf{g.}$
\end{enumerate}

Introducing $\mathbf{X}=\mathbf{e}_{\alpha }$ and $\mathbf{Y}=\mathbf{e}%
_{\beta },$ defined by (\ref{nader}), into above formulas, we compute the
N--adapted coefficients  $D=\{\Gamma _{\ \alpha \beta
}^{\gamma }\}$ and corresponding fundamental geometric objects,
\begin{eqnarray*}
\mathcal{T} &=&\{T_{\ \alpha \beta }^{\gamma }=\left( T_{\ jk}^{i},T_{\
ja}^{i},T_{\ ji}^{a},T_{\ bi}^{a},T_{\ bc}^{a}\right) \}; \\
\mathcal{R} &\mathbf{=}&\mathbf{\{}R_{\ \beta \gamma \delta }^{\alpha }%
\mathbf{=}\left( R_{\ hjk}^{i}\mathbf{,}R_{\ bjk}^{a}\mathbf{,}R_{\ hja}^{i}%
\mathbf{,}R_{\ bja}^{c}\mathbf{,}R_{\ hba}^{i},R_{\ bea}^{c}\right) \mathbf{%
\};}\ \mathcal{Q} = \mathbf{\{}Q_{\ \alpha \beta }^{\gamma }\}.
\end{eqnarray*}

Every (pseudo) Riemannian manifold $\left( V,\mathbf{g}\right) $ is
naturally equipped with a Levi--Civita connection $D=\nabla =\{\ ^{\nabla
}\Gamma _{\ \alpha \beta }^{\gamma }\}$ completely defined by $\mathbf{g}%
=\{g_{\alpha \beta }\}$ if and only if  the metric
compatibility, i.e. $\ ^{\nabla }\mathcal{Q}(X)= \nabla _{\mathbf{X}}\mathbf{g}%
=0, $ and zero torsion, i.e. $\ ^{\nabla }\mathcal{T}=0,$ conditions are satisfied. Hereafter,
we shall write, for simplicity, $\ ^{\nabla }\Gamma _{\ \alpha \beta
}^{\gamma }=\Gamma _{\ \alpha \beta }^{\gamma }.$ It should be emphasized
that $\nabla $ does not preserve under parallelism and general
frame/coordinate transforms a N--splitting (\ref{ncon}). Nevertheless, it is
possible to construct a unique distortion relation
\begin{equation}
\nabla =\widehat{\mathbf{D}}+\widehat{\mathbf{Z}},  \label{distrel}
\end{equation}%
where both linear connections $\nabla$ and $\widehat{\mathbf{D}}$ (the
second one can be considered as an auxiliary linear connection, which in
literature is called the canonical distinguished connection; in brief,
d--connection) and the distortion tensor $\widehat{\mathbf{Z}}$, i.e. all
values in the above formula, are completely defined by $\mathbf{g}%
=\{g_{\alpha \beta }\}$ for a prescribed $\mathbf{N}=\{N_{i}^{a}\},$ see
details in \cite{ijgmmp,vexsol1,vexsol2}.

\begin{theorem}
\label{thadist}With respect to N--adapted frames (\ref{nader}) and (\ref%
{nadif}), the coefficient of distortion relation (\ref{distrel}) are
computed
\begin{equation}
\Gamma _{\ \alpha \beta }^{\gamma }=\widehat{\mathbf{\Gamma }}_{\ \alpha
\beta }^{\gamma }+\widehat{\mathbf{Z}}_{\ \alpha \beta }^{\gamma },
\label{distrel1}
\end{equation}%
where the canonical d--connection $\widehat{\mathbf{D}}=\{$ $\widehat{%
\mathbf{\Gamma }}_{\ \alpha \beta }^{\gamma }=(\widehat{L}_{jk}^{i},\widehat{%
L}_{bk}^{a},\widehat{C}_{jc}^{i},\widehat{C}_{bc}^{a})\}$ is defined by
coefficients
\begin{eqnarray}
\widehat{L}_{jk}^{i} &=&\frac{1}{2}g^{ir}\left( \mathbf{e}_{k}g_{jr}+\mathbf{%
e}_{j}g_{kr}-\mathbf{e}_{r}g_{jk}\right),   \widehat{L}_{bk}^{a} =e_{b}(N_{k}^{a})+\frac{1}{2}g^{ac}\left( \mathbf{e}%
_{k}g_{bc}-g_{dc}\ e_{b}N_{k}^{d}-g_{db}\ e_{c}N_{k}^{d}\right) ,  \notag \\
\widehat{C}_{jc}^{i} &=&\frac{1}{2}g^{ik}e_{c}g_{jk},\ \widehat{C}_{bc}^{a}=%
\frac{1}{2}g^{ad}\left( e_{c}g_{bd}+e_{b}g_{cd}-e_{d}g_{bc}\right) ,  \label{cdc}
\end{eqnarray}%
and the distortion tensor $\widehat{\mathbf{Z}}_{\ \alpha \beta }^{\gamma }$
is
\begin{eqnarray}
\ Z_{jk}^{a} &=&-\widehat{C}_{jb}^{i}g_{ik}g^{ab}-\frac{1}{2}\Omega
_{jk}^{a},~Z_{bk}^{i}=\frac{1}{2}\Omega _{jk}^{c}g_{cb}g^{ji}-\Xi _{jk}^{ih}~%
\widehat{C}_{hb}^{j},  Z_{bk}^{a} =\ ^{+}\Xi _{cd}^{ab}~\widehat{T}_{kb}^{c}, \ Z_{jk}^{i}=0, \label{deft}   \\ Z_{kb}^{i}&=&\frac{1}{%
2}\Omega _{jk}^{a}g_{cb}g^{ji}+\Xi _{jk}^{ih}~\widehat{C}_{hb}^{j},
\ Z_{jb}^{a} =\ - \ ^{-}\Xi _{cb}^{ad}~\widehat{T}_{jd}^{c},\ Z_{bc}^{a}=0,\
Z_{ab}^{i}=-\frac{g^{ij}}{2}\left[ \widehat{T}_{ja}^{c}g_{cb}+\widehat{T}%
_{jb}^{c}g_{ca}\right],  \notag
\end{eqnarray}%
for $\ \Xi _{jk}^{ih}=\frac{1}{2}(\delta _{j}^{i}\delta
_{k}^{h}-g_{jk}g^{ih})$ and $~^{\pm }\Xi _{cd}^{ab}=\frac{1}{2}(\delta
_{c}^{a}\delta _{d}^{b}+g_{cd}g^{ab}).$ The nontrivial coefficients $\Omega
_{jk}^{a}$ and $\widehat{\mathbf{T}}_{\ \alpha \beta }^{\gamma }$ are given,
respectively, by $W^\alpha_{\beta \gamma}$ as formulas (\ref{nonholr}) and, see below, (\ref{dtors}).
\end{theorem}

\begin{proof}
It follows from a straightforward verification in N--adapted frames that the
sums of coefficients (\ref{cdc}) and (\ref{deft}) result in the coefficients
of the Levi--Civita connection $\Gamma _{\ \alpha \beta }^{\gamma }$ for a
general metric parametrized as a d--metric $\mathbf{g}=[g_{ij},g_{ab}]$ \ (%
\ref{dm}). $\square $
\end{proof}

\vskip5pt

All geometric constructions and physical theories derived for geometric data
$\left( \mathbf{g,}\nabla \right) $ can be equivalently modeled by geometric
data $\left( \mathbf{g,N,}\widehat{\mathbf{D}}\right) $ because of unique
distortion relation (\ref{distrel}).

\begin{theorem}
\label{thadtors}The nonholonomically induced torsion $\widehat{\mathcal{T}}$
$=\{\widehat{\mathbf{T}}_{\ \alpha \beta }^{\gamma }\}$ of $\widehat{\mathbf{%
D}}$ \ is determined in a unique form by the metric compatibility condition,
$\widehat{\mathbf{D}}\mathbf{g}=0,$ and zero horizontal and vertical torsion
coefficients, $\widehat{T}_{\ jk}^{i}=0$ and $\widehat{T}_{\ bc}^{a}=0,$ but
with nontrivial h--v-- coefficients {\small
\begin{equation}
\widehat{T}_{\ jk}^{i}=\widehat{L}_{jk}^{i}-\widehat{L}_{kj}^{i},\widehat{T}%
_{\ ja}^{i}=\widehat{C}_{jb}^{i},\widehat{T}_{\ ji}^{a}=-\Omega _{\ ji}^{a},%
\widehat{T}_{aj}^{c}=\widehat{L}_{aj}^{c}-e_{a}(N_{j}^{c}),\widehat{T}_{\
bc}^{a}=\ \widehat{C}_{bc}^{a}-\ \widehat{C}_{cb}^{a}.  \label{dtors}
\end{equation}%
}
\end{theorem}

\begin{proof}
The coefficients (\ref{dtors}) are computed by introducing $\ D=$ $\widehat{%
\mathbf{D}},$ with coefficients (\ref{cdc}), and $X=\mathbf{e}_{\alpha },Y=%
\mathbf{e}_{\beta }$ (for N--adapted frames (\ref{nader})) into standard
formula for torsion, $\mathcal{T}(X,Y):=D_{\mathbf{X}}Y-D_{\mathbf{Y}%
}X-[X,Y] $. \ $\square $
\end{proof}

\vskip5pt

In a similar form, introducing $\widehat{\mathbf{D}}$ and $X=\mathbf{e}%
_{\alpha },Y=\mathbf{e}_{\beta },Z=$ $\mathbf{e}_{\gamma }$ into \newline
$\mathcal{R}(X,Y):=D_{\mathbf{X}}D_{\mathbf{Y}}-D_{\mathbf{Y}}D_{\mathbf{X}%
}-D_{\mathbf{[X,Y]}},$ we prove

\begin{theorem}
\label{thasa1}The curvature $\widehat{\mathcal{R}}=\{\widehat{\mathbf{R}}_{\
\beta \gamma \delta }^{\alpha }\}$ of $\ \widehat{\mathbf{D}}$ is
characterized by N--adapted coefficients {\small
\begin{eqnarray}
\widehat{R}_{\ hjk}^{i} &=&e_{k}\widehat{L}_{\ hj}^{i}-e_{j}\widehat{L}_{\
hk}^{i}+\widehat{L}_{\ hj}^{m}\widehat{L}_{\ mk}^{i}-\widehat{L}_{\ hk}^{m}%
\widehat{L}_{\ mj}^{i}-\widehat{C}_{\ ha}^{i}\Omega _{\ kj}^{a},  \notag \\
\widehat{R}_{\ bjk}^{a} &=&e_{k}\widehat{L}_{\ bj}^{a}-e_{j}\widehat{L}_{\
bk}^{a}+\widehat{L}_{\ bj}^{c}\widehat{L}_{\ ck}^{a}-\widehat{L}_{\ bk}^{c}%
\widehat{L}_{\ cj}^{a}-\widehat{C}_{\ bc}^{a}\Omega _{\ kj}^{c},
\label{dcurv} \\
\widehat{R}_{\ jka}^{i} &=&e_{a}\widehat{L}_{\ jk}^{i}-\widehat{D}_{k}%
\widehat{C}_{\ ja}^{i}+\widehat{C}_{\ jb}^{i}\widehat{T}_{\ ka}^{b},\widehat{%
R}_{\ bka}^{c}=e_{a}\widehat{L}_{\ bk}^{c}-D_{k}\widehat{C}_{\ ba}^{c}+%
\widehat{C}_{\ bd}^{c}\widehat{T}_{\ ka}^{c},  \notag \\
\widehat{R}_{\ jbc}^{i} &=&e_{c}\widehat{C}_{\ jb}^{i}-e_{b}\widehat{C}_{\
jc}^{i}+\widehat{C}_{\ jb}^{h}\widehat{C}_{\ hc}^{i}-\widehat{C}_{\ jc}^{h}%
\widehat{C}_{\ hb}^{i},\  \widehat{R}_{\ bcd}^{a} = e_{d}\widehat{C}_{\ bc}^{a}-e_{c}\widehat{C}_{\
bd}^{a}+\widehat{C}_{\ bc}^{e}\widehat{C}_{\ ed}^{a}-\widehat{C}_{\ bd}^{e}%
\widehat{C}_{\ ec}^{a}.  \notag
\end{eqnarray}%
}
\end{theorem}

We can re--define the differential geometry of a (pseudo) Riemannian space $%
\mathbf{V}$ in nonholonomic form in terms of geometric data $(\mathbf{g,}%
\widehat{\mathbf{D}})$, which is equivalent to the  formulation
with $(\mathbf{g,\nabla }).$

\begin{corollary}
\label{corolricci} The Ricci tensor $\widehat{\mathbf{R}}_{\alpha \beta }:=%
\widehat{\mathbf{R}}_{\ \alpha \beta \gamma }^{\gamma }$ (\ref{riccid}) of $%
\widehat{\mathbf{D}}$ is characterized by  coefficients
\begin{equation}
\widehat{\mathbf{R}}_{\alpha \beta }=\{\widehat{R}_{ij}:=\widehat{R}_{\
ijk}^{k},\ \widehat{R}_{ia}:=-\widehat{R}_{\ ika}^{k},\ \widehat{R}_{ai}:=%
\widehat{R}_{\ aib}^{b},\ \widehat{R}_{ab}:=\widehat{R}_{\ abc}^{c}\}.
\label{driccic}
\end{equation}
\end{corollary}

\begin{proof}
The formulas for $h$--$v$--components (\ref{driccic}) are obtained by
contracting, respectively, the coefficients (\ref{dcurv}). Using $\widehat{%
\mathbf{D}}$ (\ref{cdc}), we express such formulas in terms of partial
derivatives of coefficients of metric $\mathbf{g}$ (\ref{mst}) and any
equivalent parametrization in the form (\ref{dm}), or (\ref{ansatz}). $%
\square $
\end{proof}

\vskip5pt

The scalar curvature $\ ^{s}\widehat{R}$ of $\ \widehat{\mathbf{D}}$ is by
definition
\begin{equation}
\ ^{s}\widehat{R}:=\mathbf{g}^{\alpha \beta }\widehat{\mathbf{R}}_{\alpha
\beta }=g^{ij}\widehat{R}_{ij}+g^{ab}\widehat{R}_{ab}.  \label{sdcurv}
\end{equation}%
In order to elaborate models of gravity theories for $\nabla $ and/or $%
\widehat{\mathbf{D}}$, we have to consider the corresponding Ricci tensors,%
\begin{eqnarray}
Ric=\{R_{\ \beta \gamma } &:=&R_{\ \beta \gamma \alpha }^{\alpha }\},%
\mbox{ for
}\nabla =\{\Gamma _{\ \alpha \beta }^{\gamma }\},  \label{riccie} \\
\mbox{ and }\widehat{R}ic &=&\{\widehat{\mathbf{R}}_{\ \beta \gamma }:=%
\widehat{\mathbf{R}}_{\ \beta \gamma \alpha }^{\alpha }\},\mbox{ for }%
\widehat{\mathbf{D}}=\{\widehat{\mathbf{\Gamma }}_{\ \alpha \beta }^{\gamma
}\}.  \label{riccid}
\end{eqnarray}%
For instance, using (\ref{driccic}) and (\ref{sdcurv}), we can compute the
Einstein tensor $\widehat{\mathbf{E}}_{\alpha \beta }$ of $\widehat{\mathbf{D%
}},$
\begin{equation}
\widehat{\mathbf{E}}_{\alpha \beta }\doteqdot \widehat{\mathbf{R}}_{\alpha
\beta }-\frac{1}{2}\mathbf{g}_{\alpha \beta }\ ^{s}\widehat{R}.
\label{enstdt}
\end{equation}%
In general, this tensor is different from a similar one constructed with (%
\ref{riccie}) for the Levi--Civita connection $\nabla .$

\begin{proposition}
\label{approp}The N--adapted coefficients $\widehat{\mathbf{\Gamma }}_{\
\alpha \beta }^{\gamma }$ of $\widehat{\mathbf{D}}$ are identic to the
coefficients $\Gamma _{\ \alpha \beta }^{\gamma }$ of $\ \nabla ,$ both sets
being computed with respect to N--adapted frames (\ref{nader}) and (\ref%
{nadif}), if and only if the conditions $\widehat{L}%
_{aj}^{c}=e_{a}(N_{j}^{c}),\widehat{C}_{jb}^{i}=0$ and $\Omega _{\
ji}^{a}=0$ are satisfied.
\end{proposition}

\begin{proof}
If the conditions of the Proposition, i.e., constraints (\ref{lcconstr}), are
satisfied, all N--adapted coefficients of the torsion $\widehat{\mathbf{T}}%
_{\ \alpha \beta }^{\gamma }$ \ (\ref{dtors}) are zero. In such a case, the
distortion tensor $\widehat{\mathbf{Z}}_{\ \alpha \beta }^{\gamma }$ (\ref%
{deft}) is also zero. Following formula (\ref{distrel1}), we get $\Gamma _{\
\alpha \beta }^{\gamma }= \widehat{\mathbf{\Gamma }}_{\ \alpha \beta
}^{\gamma }.$ Inversely, if the last equalities of coefficients are
satisfied for a chosen splitting (\ref{ncon}), we get trivial torsions and
distortions of $\nabla .$ We emphasize that $\widehat{\mathbf{D}}\neq \nabla
$ because such connections have different transformation rules under
frame/coordinate transforms. Nevertheless, \ if $\Gamma _{\ \alpha \beta
}^{\gamma }=\widehat{\mathbf{\Gamma }}_{\ \alpha \beta }^{\gamma }$ in a
N--adapted frame of reference, we get corresponding equalities for the
Riemann and Ricci tensors etc. This means that the N--coefficients are such
way fixed via frame transforms that the nonholonomic distribution became
integrable even, in general, the frames (\ref{nader}) and (\ref{nadif}) are
nonholonomic (because not all anholonomy coefficients are not obligatory
zero, for instance, $w_{ia}^{b}=\partial _{a}N_{i}^{b}$ may be nontrivial. $\square $
\end{proof}

\subsection{Field equations in nonholonomically modified gravity}

We study modified gravity theories derived for the action%
\begin{equation}
S=\frac{1}{16\pi }\int \delta u^{4}\sqrt{|\mathbf{g}_{\alpha \beta }|}\left(
f(\ ^{s}\widehat{R},T)+\ ^{m}L\right)  \label{act}
\end{equation}%
where $\ ^{m}L$ is the matter Lagrangian density for which the
stress--energy tensor of matter is defined via variation on inverse metric
tensor as $\mathbf{T}_{\alpha \beta }=-\frac{2}{\sqrt{|\mathbf{g}_{\mu \nu }|}}\frac{%
\delta (\sqrt{|\mathbf{g}_{\mu \nu }|}\ ^{m}L)}{\delta \mathbf{g}^{\alpha \beta }}$,
trace $T:=\mathbf{g}^{\alpha \beta }\mathbf{T}_{\alpha \beta },$ and $f(\
^{s}\widehat{R},T)$ is an arbitrary functional on $\ ^{s}\widehat{R}$ (\ref%
{sdcurv}) and $T.$ For simplicity, we assume that the stress--energy tensor
of the matter is given by
\begin{equation}
\mathbf{T}_{\alpha \beta }=(\rho +p)\mathbf{v}_{\alpha }\mathbf{v}_{\beta }-p%
\mathbf{g}_{\alpha \beta },  \label{emt}
\end{equation}%
where in the approximation of perfect fluid matter $\rho $ is the energy
density, $p$ is the pressure and the four--velocity $\mathbf{v}_{\alpha }$
is subjected to the conditions $\mathbf{v}_{\alpha }\mathbf{v}^{\alpha }=1$
and $\mathbf{v}^{\alpha }\widehat{\mathbf{D}}_{\beta }\mathbf{v}_{\alpha
}=0, $ for $\ ^{m}L=-p$ in a corresponding local frame. We also consider
approximations for type%
\begin{equation}
f(\ ^{s}\widehat{R},T)=\ ^{1}f(\ ^{s}\widehat{R})+\ ^{2}f(T)  \label{functs}
\end{equation}%
and denote by $\ ^{1}F(\ ^{s}\widehat{R}):=\partial \ ^{1}f(\ ^{s}\widehat{R}%
)/\partial \ ^{s}\widehat{R}$ and $\ ^{2}F(T):=\partial \ ^{2}f(T)/\partial
\ T.$

\begin{theorem}
The gravitational field equations for a modified gravity model (\ref{act})
with $f$--functional (\ref{functs}) and perfect fluid stress--energy tensor (%
\ref{emt}) can be re--written equivalently using the canonical d--connection
$\widehat{\mathbf{D}},$
\begin{equation}
\widehat{\mathbf{R}}_{\ \beta \delta }-\frac{1}{2}\mathbf{g}_{\beta \delta
}\ ^{s}R=\mathbf{\Upsilon }_{\beta \delta },  \label{cdeinst}
\end{equation}%
where the source d--tensor $\mathbf{\Upsilon }_{\beta \delta }$ is such way
constructed that $\mathbf{\Upsilon }_{\beta \delta }\rightarrow 8\pi
GT_{\beta \delta }$ \ for $\widehat{\mathbf{D}}\rightarrow \nabla ,$ where $%
T_{\beta \delta }$ is the energy--momentum tensor in GR with coupling
gravitational constant $G.$ In explicit form,
\begin{equation*}
\mathbf{\Upsilon }_{\beta \delta }=\ ^{ef}\eta\ G\ \mathbf{T}_{\beta \delta
}+\ ^{ef}\mathbf{T}_{\beta \delta },
\end{equation*}%
where the effective polarization of cosmological constant is computed
 $\ ^{ef}\eta =[1+\ ^{2}F/8\pi ]/\ ^{1}F$ and the $f$--modification of the energy--momentum tensor is computed as an
additional effective source
 $\ ^{ef}\mathbf{T}_{\beta \delta }=[\frac{1}{2}(\ ^{1}f-\ ^{1}F\ ^{s}\widehat{%
R}+2p\ ^{2}F+\ ^{2}f)\mathbf{g}_{\beta \delta }-(\mathbf{g}_{\beta \delta }\
\widehat{\mathbf{D}}_{\alpha }\widehat{\mathbf{D}}^{\alpha }-\widehat{%
\mathbf{D}}_{\beta }\widehat{\mathbf{D}}_{\delta })\ ^{1}F]/\ ^{1}F$.
 \end{theorem}

\begin{proof}
By varying the action $S$ (\ref{act}) with respect to $\mathbf{g}^{\alpha
\beta }$ and following a N--adapted covariand differential calculus with
respect bases (\ref{nader}) and (\ref{nadif}) (see similar results for the
Levi--Civita connection in \cite{odintsov}) we obtain the gravitational
field equations (\ref{cdeinst}) and respective effective sources. $\square $
\end{proof}

\vskip5pt

We consider matter field sources in (\ref{cdeinst}) which can be
diagonalized with respect to N--adapted frames, {\small
\begin{equation}
\mathbf{\Upsilon }_{~\delta }^{\beta }=diag[\mathbf{\Upsilon }_{\alpha }:%
\mathbf{\Upsilon }_{~1}^{1}=\mathbf{\Upsilon }_{~2}^{2}=\Upsilon
(x^{k},y^{3})+\underline{\Upsilon }(x^{k},y^{4});\mathbf{\Upsilon }_{~3}^{3}=%
\mathbf{\Upsilon }_{~4}^{4}=~^{v}\Upsilon (x^{k})].  \label{source}
\end{equation}%
} Such a formal diagonalization can be performed via corresponding frame/
coordinate transforms for very general distributions of matter fields. Such
effective sources can be considered as nonholonomic constraints on the Ricci
tensor (see Theorem \ref{th2a}) computed for certain general assumptions on
modified off--diagonal gravitational interactions.

\begin{corollary}
The gravitational field equations (\ref{cdeinst}) transform into the
Einstein equations in GR, in "standard" form for $\nabla ,$
\begin{equation}
E_{\beta \delta }=R_{\beta \delta }-\frac{1}{2}\mathbf{g}_{\beta \delta }\
R=\varkappa T_{\beta \delta },  \label{einsteqs}
\end{equation}%
where $R:=\mathbf{g}^{\beta \delta }R_{\ \beta \delta },$ if $\ ^{2}f=0,$ \ $%
\ ^{1}f(\ ^{s}\widehat{R})=R,$ for the same N--adapted coefficients for $%
\widehat{\mathbf{D}}$ and $\nabla $ if\
\begin{equation}
\widehat{L}_{aj}^{c}=e_{a}(N_{j}^{c}),\ \widehat{C}_{jb}^{i}=0,\ \Omega _{\
ji}^{a}=0.  \label{lcconstr}
\end{equation}
\end{corollary}

\begin{proof}
In general, the systems of PDEs (\ref{cdeinst}) and (\ref{einsteqs}) are very
different. But if the constraints (\ref{lcconstr}) are imposed additionally
on $\widehat{\mathbf{D}},$ we satisfy the conditions of Proposition \ref%
{approp}, when $\Gamma _{\ \alpha \beta }^{\gamma }=\widehat{\mathbf{\Gamma }%
}_{\ \alpha \beta }^{\gamma }$ results in $R_{\beta \delta }=\widehat{%
\mathbf{R}}_{\ \beta \delta }$ and $E_{\alpha \beta }=\widehat{\mathbf{E}}%
_{\alpha \beta }.$ $\square $
\end{proof}

\section{Decoupling of Gravitational Field Equations}

\label{s3} The Einstein equations and their generalizations to various models of
commutative and/or noncommutative gravity (for instance, in Finsler
variables) possess a very important property that they
decouple with respect to certain N--adapted frames of reference (\ref{nader}%
) and (\ref{nadif}) and for the canonical d--connection $\widehat{\mathbf{D}}
$. The aim of this section is to show how the anholonomic frame deformation
method \cite{ijgmmp,vexsol1,vexsol2}, AFDM, can be applied to decouple the
fundamental field equations (\ref{cdeinst}) \ for modified gravity.

\subsection{Off--diagonal spacetimes with Killing symmetry}

We use the ansatz (\ref{ans1}) when $\omega =1,\underline{h}_{3}=1,\underline{%
w}_{i}=0$ and $\underline{n}_{i}=0$ in data (\ref{paramdcoef}) and $%
\underline{\mathbf{\Upsilon }}=0$ for (\ref{source}). Such a generic
off--diagonal metric does not depend on variable $y^{4},$ i.e., $\partial
/\partial y^{4}$ is a Killing vector, if $\underline{h}_{4}=1.$
Nevertheless, the decoupling property can proven for the same assumptions
but arbitrary $\underline{h}_{4}(x^{k},y^{4})$ with nontrivial dependence on
$y^{4}.$ We call this class of metrics to be with weak Killing symmetry
because they result in systems of PDEs (\ref{cdeinst}) as for the Killing
case but there are differences in (\ref{lcconstr}) if $\underline{h}_{4}\neq
1.$ It will be convenient to use brief denotations for partial derivatives, $%
a^{\bullet }=\partial a/\partial x^{1},a^{\prime }=\partial a/\partial
x^{2},a^{\ast }=\partial a/\partial y^{3},a^{\circ }=\partial a/\partial
y^{4}.$ The equations will be written with respect to N--adapted frames of
type (\ref{nader}) and (\ref{nadif}).

\begin{theorem}
\label{th2a}The gravitational field equations (\ref{cdeinst}) with possible
constraints (\ref{lcconstr}) for a metric $\mathbf{g}$ (\ref{paramdcoef})
with $\omega =\underline{h}_{3}=1$ and $\underline{w}_{i}=\underline{n}%
_{i}=0 $ and $\underline{\mathbf{\Upsilon }}=0$ in matter source $\mathbf{%
\Upsilon }_{~\delta }^{\beta }$ (\ref{source}) are equivalent, respectively,
to
\begin{eqnarray}
\widehat{R}_{1}^{1} &=&\widehat{R}_{2}^{2}=\frac{-1}{2g_{1}g_{2}}%
[g_{2}^{\bullet \bullet }-\frac{g_{1}^{\bullet }g_{2}^{\bullet }}{2g_{1}}-%
\frac{\left( g_{2}^{\bullet }\right) ^{2}}{2g_{2}}+g_{1}^{\prime \prime }-%
\frac{g_{1}^{\prime }g_{2}^{\prime }}{2g_{2}}-\frac{(g_{1}^{\prime })^{2}}{%
2g_{1}}]=-\ ^{v}\Upsilon ,  \label{eq1b} \\
\widehat{R}_{3}^{3} &=&\widehat{R}_{4}^{4}=-\frac{1}{2h_{3}h_{4}}%
[h_{4}^{\ast \ast }-\frac{\left( h_{4}^{\ast }\right) ^{2}}{2h_{4}}-\frac{%
h_{3}^{\ast }h_{4}^{\ast }}{2h_{3}}]=-\Upsilon ,  \label{eq2b} \\
\widehat{R}_{3k} &=&\frac{w_{k}}{2h_{4}}[h_{4}^{\ast \ast }-\frac{\left(
h_{4}^{\ast }\right) ^{2}}{2h_{4}}-\frac{h_{3}^{\ast }h_{4}^{\ast }}{2h_{3}}%
]+\frac{h_{4}^{\ast }}{4h_{4}}\left( \frac{\partial _{k}h_{3}}{h_{3}}+\frac{%
\partial _{k}h_{4}}{h_{4}}\right) -\frac{\partial _{k}h_{4}^{\ast }}{2h_{4}}%
=0,  \label{eq3b} \\
\widehat{R}_{4k} &=&\frac{h_{4}}{2h_{3}}n_{k}^{\ast \ast }+\left( \frac{h_{4}%
}{h_{3}}h_{3}^{\ast }-\frac{3}{2}h_{4}^{\ast }\right) \frac{n_{k}^{\ast }}{%
2h_{3}}=0,  \label{eq4b}
\end{eqnarray}%
\begin{eqnarray}
\mbox{ and }\ w_{i}^{\ast } &=&(\partial _{i}-w_{i}\partial _{3})\ln |h_{4}|,\partial
_{k}w_{i}=\partial _{i}w_{k},  \label{lccond} \\
n_{k}\underline{h}_{4}^{\circ } &=&\mathbf{\partial }_{k}\underline{h}%
_{4},n_{i}^{\ast }=0,\partial _{i}n_{k}=\partial _{k}n_{i}.  \notag
\end{eqnarray}%
\end{theorem}

\begin{proof}
See  \ref{sb}. $\square $
\end{proof}

\vskip5pt

Let us discuss the decoupling (splitting) property of the modified and
Einstein equations with respect to certain classes of N--adapted frames
which is contained in the system of PDEs (\ref{eq1b})-- (\ref{lccond}). For
instance, the first equation is for a 2--d metric which always can be
diagonalized, $[g_{1},g_{2}],$ and/or made to be conformally flat.
Prescribing a function $g_{1}$ and source $\ ^{v}\Upsilon $, we can find $%
g_{2},$ or inversely. The equation (\ref{eq2b}) contains only the first and
second derivatives on $\partial /\partial y^{3}$ and relates two functions $%
h_{3}$ and $h_{4}.$ Prescribing one of such functions and source $\Upsilon ,$
we can define the second one taking, respectively, one or two derivations on
$y^{3}.$ The equation (\ref{eq3b}) is a linear algebraic system for $w_{k}$
if the coefficients $h_{a}$ have been already defined as a solution of (\ref%
{eq2b}). Nevertheless, we have to solve a system of first order PDE on $%
x^{k} $ and $y^{3}$ in order to find $w_{k}$ resulting in zero torsion
conditions (\ref{lccond}). Such conditions do not allow \ a complete
decoupling because the first equations relate $w_{i}$ to $H=\ln |h_{4}|$ via
corresponding first order PDEs. Nevertheless, it is possible, for instance,
to integrate such solutions for any prescribed $H$ (see more details below,
in Remark  \ref{remex}). The fourth equations (\ref{eq4b}) became
trivial for any $n_{i}^{\ast }=0$ if we wont to satisfy completely such zero
torsion conditions\footnote{%
Nontrivial solutions and nonzero torsion configurations present interest in
modified theories of gravity and brane physics, see such examples in Ref.
\cite{vsing2}.}. A nontrivial function $\underline{h}_{4}$ is explicitly
present in the conditions (\ref{lccond}). If such restrictions are
satisfied, this allows us to eliminate $\underline{h}_{4}$ from (\ref{eq3b}).

We conclude that the modified equations for metrics with one Killing
symmetry can be in such way parametrized with respect to N--adapted frames that
they decouple and separate into "quite simple" PDE for h--components, $g_{i}$, and then for
v--components, $h_{a}$. The N--connection coefficients also separate and can
be defined from corresponding algebraic and/or first order PDE. The
 "zero torsion" conditions impose certain
additional constraints (as some simple first order PDE with possible
separation of variables) on N--coefficients and coefficients of v--metric.

In a similar form, we can decouple the modified gravitational field equations for
spacetimes with one Killing symmetry on $\partial /\partial y_{4}.$

\begin{corollary}
\label{corol1} The equations (\ref{cdeinst}) and (\ref{lcconstr}) for a
metric $\mathbf{g}$ (\ref{paramdcoef}) with $\omega =h_{4}=1$ and $%
w_{i}=n_{i}=0$ and $\mathbf{\Upsilon }=0$ in matter source $\mathbf{\Upsilon
}_{~\delta }^{\beta }$ (\ref{source}) are equivalent, respectively, to
\begin{eqnarray}
\widehat{R}_{1}^{1} &=&\widehat{R}_{2}^{2}=\frac{-1}{2g_{1}g_{2}}%
[g_{2}^{\bullet \bullet }-\frac{g_{1}^{\bullet }g_{2}^{\bullet }}{2g_{1}}-%
\frac{\left( g_{2}^{\bullet }\right) ^{2}}{2g_{2}}+g_{1}^{\prime \prime }-%
\frac{g_{1}^{\prime }g_{2}^{\prime }}{2g_{2}}-\frac{(g_{1}^{\prime })^{2}}{%
2g_{1}}]=-\ ^{v}\Upsilon ,  \label{eq1c} \\
\widehat{R}_{3}^{3} &=&\widehat{R}_{4}^{4}=-\frac{1}{2\underline{h}_{3}%
\underline{h}_{4}}[\underline{h}_{3}^{\circ \circ }-\frac{\left( \underline{h%
}_{3}^{\circ }\right) ^{2}}{2\underline{h}_{3}}-\frac{\underline{h}%
_{3}^{\circ }\underline{h}_{4}^{\circ }}{2\underline{h}_{4}}]=-\underline{%
\Upsilon },  \label{eq2c} \\
\widehat{R}_{3k} &=&+\frac{\underline{h}_{3}}{2\underline{h}_{4}}\underline{w%
}_{k}^{\circ \circ }+\left( \frac{\underline{h}_{3}}{\underline{h}_{4}}%
\underline{h}_{4}^{\circ }-\frac{3}{2}\underline{h}_{3}^{\circ }\right)
\frac{\underline{h}_{k}^{\circ }}{2\underline{h}_{4}}=0,  \label{eq3c} \\
\widehat{R}_{4k} &=&\frac{\underline{n}_{k}}{2\underline{h}_{3}}[\underline{h%
}_{3}^{\circ \circ }-\frac{\left( \underline{h}_{3}^{\circ }\right) ^{2}}{2%
\underline{h}_{3}}-\frac{\underline{h}_{3}^{\circ }\underline{h}_{4}^{\circ }%
}{2\underline{h}_{4}}]+\frac{\underline{h}_{3}^{\circ }}{4\underline{h}_{3}}%
\left( \frac{\partial _{k}\underline{h}_{3}}{\underline{h}_{3}}+\frac{%
\partial _{k}\underline{h}_{4}}{\underline{h}_{4}}\right) -\frac{\partial
_{k}h_{3}^{\circ }}{2h_{3}}=0,  \label{eq4c}
\end{eqnarray}%
\begin{eqnarray}
\mbox{ \ and \ }\underline{n}_{i}^{\circ } &=&(\partial _{i}-\underline{n}%
_{i}\partial _{4})\ln |\underline{h}_{3}|,(\partial _{k}-\underline{n}%
_{k}\partial _{4})\underline{n}_{i}=(\partial _{i}-\underline{n}_{i}\partial
_{4})\underline{n}_{k},  \label{lccondd} \\
\underline{w}_{k}h_{3}^{\ast } &=&\mathbf{\partial }_{k}h_{3},\underline{w}%
_{i}^{\circ }=0,\partial _{i}\underline{w}_{k}=\partial _{k}\underline{w}%
_{i}.  \notag
\end{eqnarray}%
\end{corollary}

\begin{proof}
It is similar to that for Theorem \ref{th2a} provided in Appendix \ref{sb}.
We do not repeat such computations. $\square $
\end{proof}

\vskip5pt

Using the above Theorem and Corollary, we can state:

\begin{conclusion}
\label{concl1} The nonlinear systems of PDEs corresponding to (modified)
gravitational equations (\ref{cdeinst}) and (\ref{lcconstr}) for metrics $%
\mathbf{g}$ (\ref{paramdcoef}) with Killing symmetry on $\partial /\partial
y_{4},$ when $\omega =\underline{h}_{3}=1$ and $\underline{w}_{i}=\underline{%
n}_{i}=0$ and $\underline{\mathbf{\Upsilon }}=0$ in matter source $\mathbf{%
\Upsilon }_{~\delta }^{\beta }$ (\ref{source}), can be transformed into
respective systems of PDEs for data with Killing symmetry on $\partial
/\partial y_{3},$ when $\omega =h_{4}=1$ and $w_{i}=n_{i}=0$ and $\mathbf{%
\Upsilon }=0$, if $h_{3}(x^{i},y^{3})\rightarrow \underline{h}%
_{4}(x^{i},y^{4}),$ $h_{4}(x^{i},y^{3})\rightarrow \underline{h}%
_{3}(x^{i},y^{4}),$ $w_{k}(x^{i},y^{3})\rightarrow \underline{n}%
_{k}(x^{i},y^{4})$ and $n_{k}(x^{i},y^{3})\rightarrow \underline{w}%
_{k}(x^{i},y^{4}).$
\end{conclusion}

The above presented method of nonholonomic deformations can be used for
decoupling gravitational field equations if some metrics do not possess,
in general, any Killing symmetries. The generic nonlinear character of such
systems of PDE does not allow us to use a principle of superposition of
solutions. Nevertheless, certain classes of conformal transforms for the
v--components of d--metrics and nonholonomic constraints of integral
varieties give us the possibility to extend the anholonomic deformation
method to "non--Killing" vacuum and
nonvacuum gravitational interactions. In the next two subsections, we analyze
two possibilities to decouple the effective Einstein equations for metrics
with coefficients depending on all spacetime coordinates.

\subsection{Preserving decoupling under v--conformal transforms}

This property is stated by
\begin{lemma}
\label{lemma1}The gravitational field equations (\ref{cdeinst}) for
geometric data (\ref{data1a}), i.e. the system of PDEs (\ref{eq2b})--(\ref%
{eq4b}), do not change under a "vertical" conformal transform with
nontrivial $\omega (x^{k},y^{a})$ to a d--metric (\ref{paramdcoef}) if  the following conditions are satisfied,%
\begin{equation}
\partial _{k}\omega -w_{i}\omega ^{\ast }-n_{i}\omega ^{\circ }=0\mbox{ and }%
\text{ }\widehat{T}_{kb}^{a}=0.  \label{conf2a}
\end{equation}
\end{lemma}

\begin{proof}
It follows from straightforward computations when coefficients\newline
$g_{i}(x^{k}),g_{3}=h_{3}(x^{k},y^{3}),$ $g_{4}=h_{4}(x^{k},y^{3})\underline{h}%
_{4}(x^{k},y^{4}),N_{i}^{3}=w_{i}(x^{k},y^{3}),$ $N_{i}^{4}=n_{i}(x^{k},y^{3})$ are generalized to a nontrivial $\omega
(x^{k},y^{a})$ with $~^{\omega }g_{3}=\omega ^{2}h_{3}$ and $~^{\omega
}g_{4}=\omega ^{2}h_{4}\underline{h}_{4}.$ Using, respectively, formulas (\ref%
{distrel1}), (\ref{cdc}), (\ref{dcurv}) and (\ref{driccic}), we get
distortion relations for the Ricci d--tensors (\ref{riccid}), %
$\ ^{\omega }\widehat{R}_{~b}^{a}=\widehat{R}_{~b}^{a}+~^{\omega }\widehat{Z}%
_{~b}^{a}$ and $\ ~^{\omega }\widehat{R}_{bi}=\widehat{R}_{bi}=0$,
where $\widehat{R}_{~b}^{a}$ and $\widehat{R}_{bi}$ are those computed for $%
\omega =1,$ i.e., (\ref{eq2b})--(\ref{eq4b}). The values $~^{\omega }\widehat{%
R}_{~b}^{a}$ and$~^{\omega }\widehat{Z}_{~b}^{a}$ are defined by a
nontrivial $\omega $ and computed using the same formulas. We do not provide
certain similar details from Refs. \cite{vexsol1,vexsol2} for $\underline{h}%
_{4}=1$ because a nontrivial $\underline{h}_{4}$ does not modify
substantially the proof that $~^{\omega }\widehat{Z}_{~b}^{a}=0$ if the
conditions (\ref{conf2a}) are satisfied. $\square $
\end{proof}

\vskip5pt

Using the Theorem \ref{th2a}, Corollary \ref{corol1}, Conclusion \ref{concl1}%
, Lemma \ref{lemma1}, we prove

\begin{theorem}
\label{th2b}A d--metric
\begin{eqnarray}
\mathbf{g} &=&g_{i}(x^{k})dx^{i}\otimes dx^{i}+ \omega
^{2}(x^{k},y^{a})\left( h_{3}\mathbf{e}^{3}\otimes \mathbf{e}^{3}+h_{4}\
\underline{h}_{4}\mathbf{e}^{4}\otimes \mathbf{e}^{4}\right) ,  \notag \\
\mathbf{e}^{3} &=&dy^{3}+w_{i}(x^{k},y^{3})dx^{i},\ \mathbf{e}%
^{4}=dy^{4}+n_{i}(x^{k})dx^{i},  \label{class1}
\end{eqnarray}%
satisfying the PDE (\ref{eq1b})-- (\ref{lccond}) and $\partial _{k}\omega
-w_{i}\omega ^{\ast }-n_{i}\omega ^{\circ }=0$, or a d--metric
\begin{eqnarray}
\mathbf{g} &=&g_{i}(x^{k})dx^{i}\otimes dx^{i}+ \omega
^{2}(x^{k},y^{a})\left( h_{3} \underline{h}_{3} \mathbf{e}^{3}\otimes
\mathbf{e}^{3} +\underline{h}_{4} \mathbf{e}^{4}\otimes \mathbf{e}%
^{4}\right) ,  \notag \\
\mathbf{e}^{3} &=&dy^{3}+\underline{w}_{i}(x^{k})dx^{i},\ \mathbf{e}%
^{4}=dy^{4}+\underline{n}_{i}(x^{k},y^{4})dx^{i},  \label{class2}
\end{eqnarray}%
satisfying the PDEs (\ref{eq1c})-- (\ref{lccondd}) and $\partial _{k}\omega -%
\underline{w}_{i}\omega ^{\ast }-\underline{n}_{i}\omega ^{\circ }=0$,
define, in general, two different classes of generic off--diagonal solutions
\ of Einstein equations (\ref{cdeinst}) and (\ref{lcconstr}) with respective
sources of type (\ref{source}).
\end{theorem}

Both ansatzes of type (\ref{class1}) and (\ref{class2}) consist particular
cases of parametrizations of metrics in the form (\ref{ans1}). Via
frame/coordinate transform into a finite region of a point $~^{0}u\in
\mathbf{V}$ any spacetime metric in\ GR and $f$--modifications can be
represented in an above mentioned d--metric form. If only one of coordinates
$y^{a}$ is timelike, the solutions of type (\ref{class1}) and (\ref{class2})
can not be transformed mutually via nonholonomic frame deformations
preserving causality.

\subsection{Decoupling with effective linearization of the Ricci tensors}

The explicit form of field equations for vacuum and nonvacuum gravitational
interactions depends on the type of frames and coordinate systems we
consider for decoupling such PDEs. We can split such systems for more general
parameterizations (than ansatzes (\ref{class1}) and (\ref{class2})) in a form (%
\ref{ans1}) with nontrivial $\omega $ and N--coefficients in (\ref%
{paramdcoef}). This is possible in any open region $U\subset \mathbf{V}$
where for computing the N--adapted coefficients of the Riemann and Ricci
d--tensors, see formulas (\ref{dcurv}) and (\ref{driccic}), we can neglect
contributions from quadratic terms of type $\widehat{\Gamma }\cdot \widehat{%
\Gamma }$ but preserve values of type $\partial _{\mu }\widehat{\Gamma }.$
For such constructions, we have to introduce a class of N--adapted normal
coordinates when $\widehat{\Gamma }(u_{0})=0$ for points $u_{0}$, for
instance, belonging to a line on $U.$ Such conditions can be satisfied for
decompositions of metrics and connections on a small parameter like it is
explained in details in Ref. \cite{ijgmmp} (see decompositions on a small
eccentricity parameter $\varepsilon$ in Section \ref{s5}). Other
possibilities can be found if we impose nonholonomic constraints, for
instance, of type $h_{4}^{\ast }=0$ but for nonzero $h_{4}(x^{k},y^{3})$
and/or $h_{4}^{\ast \ast }(x^{k},y^{3})$; such constraints can be solved in
non--explicit form and define a corresponding subclass of N--adapted frames.
Considering further nonholonomic deformations with a general decoupling with
respect to a "convenient" system of reference/coordinates, we can deform
the equations and solutions to configurations when terms of type $\widehat{%
\Gamma }\cdot \widehat{\Gamma }$ became important.

\begin{theorem}["Non--quadratic" decoupling]
\label{th2}The effective Einstein equations in  modified gravity (for instance, in the form (\ref{cdeinst}) and (\ref%
{lcconstr})), via nonholonomic frame deformations to a metric $\mathbf{g}$ (%
\ref{paramdcoef}) and matter source $\mathbf{\Upsilon }_{~\delta }^{\beta }$
(\ref{source}), when contributions from terms of type $\widehat{\Gamma }%
\cdot \widehat{\Gamma }$ are considered small for an open region $U\subset
\mathbf{V,}$ can be transformed equivalently into a system of PDEs with
h--v--decoupling {\small
\begin{eqnarray}
\widehat{R}_{1}^{1}&=&\widehat{R}_{2}^{2}=\frac{-1}{2g_{1}g_{2}}%
[g_{2}^{\bullet \bullet }-\frac{g_{1}^{\bullet }g_{2}^{\bullet }}{2g_{1}}-%
\frac{\left( g_{2}^{\bullet }\right) ^{2}}{2g_{2}}+g_{1}^{\prime \prime }-%
\frac{g_{1}^{\prime }g_{2}^{\prime }}{2g_{2}}-\frac{(g_{1}^{\prime })^{2}}{%
2g_{1}}]=-\ ^{v}\Upsilon ,  \label{eq1} \\
\widehat{R}_{3}^{3} &=&\widehat{R}_{4}^{4}=-\frac{1}{2h_{3}h_{4}}%
[h_{4}^{\ast \ast }-\frac{\left( h_{4}^{\ast }\right) ^{2}}{2h_{4}}-\frac{%
h_{3}^{\ast }h_{4}^{\ast }}{2h_{3}}]-\frac{1}{2\underline{h}_{3}\underline{h}%
_{4}}[\underline{h}_{3}^{\circ \circ }-\frac{\left( \underline{h}_{3}^{\circ
}\right) ^{2}}{2\underline{h}_{3}}-\frac{\underline{h}_{3}^{\circ }%
\underline{h}_{4}^{\circ }}{2\underline{h}_{4}}]  =-\Upsilon -\underline{\Upsilon },  \label{eq2} \\
\widehat{R}_{3k} &=&\frac{w_{k}}{2h_{4}}[h_{4}^{\ast \ast }-\frac{\left(
h_{4}^{\ast }\right) ^{2}}{2h_{4}}-\frac{h_{3}^{\ast }h_{4}^{\ast }}{2h_{3}}%
]+\frac{h_{4}^{\ast }}{4h_{4}}\left( \frac{\partial _{k}h_{3}}{h_{3}}+\frac{%
\partial _{k}h_{4}}{h_{4}}\right) -\frac{\partial _{k}h_{4}^{\ast }}{2h_{4}}
\label{eq3} \\
&&+\frac{\underline{h}_{3}}{2\underline{h}_{4}}\underline{n}_{k}^{\circ
\circ }+\left( \frac{\underline{h}_{3}}{\underline{h}_{4}}\underline{h}%
_{4}^{\circ }-\frac{3}{2}\underline{h}_{3}^{\circ }\right) \frac{\underline{n%
}_{k}^{\circ }}{2\underline{h}_{4}}=0,  \notag \\
\widehat{R}_{4k} &=&\frac{\underline{w}_{k}}{2\underline{h}_{3}}[\underline{h%
}_{3}^{\circ \circ }-\frac{\left( \underline{h}_{3}^{\circ }\right) ^{2}}{2%
\underline{h}_{3}}-\frac{\underline{h}_{3}^{\circ }\underline{h}_{4}^{\circ }%
}{2\underline{h}_{4}}]+\frac{\underline{h}_{3}^{\circ }}{4\underline{h}_{3}}%
\left( \frac{\partial _{k}\underline{h}_{3}}{\underline{h}_{3}}+\frac{%
\partial _{k}\underline{h}_{4}}{\underline{h}_{4}}\right) -\frac{\partial
_{k}\underline{h}_{3}^{\circ }}{2\underline{h}_{3}}  \notag \\
&&+\frac{h_{4}}{2h_{3}}n_{k}^{\ast \ast }+\left( \frac{h_{4}}{h_{3}}%
h_{3}^{\ast }-\frac{3}{2}h_{4}^{\ast }\right) \frac{n_{k}^{\ast }}{2h_{3}}=0,
\label{eq4} \\
w_{i}^{\ast } &=&(\partial _{i}-w_{i}\partial _{3})\ln |h_{4}|,(\partial
_{k}-w_{k}\partial _{3})w_{i}=(\partial _{i}-w_{i}\partial
_{3})w_{k},n_{i}^{\ast }=0,\partial _{i}n_{k}=\partial _{k}n_{i}, \notag
\label{lcconstr1} \\
\underline{w}_{i}^{\circ } &=&0,\ \partial _{i}\underline{w}_{k}=\partial
_{k}\underline{w}_{i},\underline{n}_{i}^{\circ }=(\partial _{i}-\underline{n}%
_{i}\partial _{4})\ln |\underline{h}_{3}|,(\partial _{k}-\underline{n}%
_{k}\partial _{4})\underline{n}_{i}=(\partial _{i}-\underline{n}_{i}\partial
_{4})\underline{n}_{k},\   \notag \\
\mathbf{e}_{k}\omega &=&\partial _{k}\omega -\left( w_{i}+\underline{w}%
_{i}\right) \omega ^{\ast }-\left( n_{i}+\underline{n}_{i}\right) \omega
^{\circ }=0.  \label{conf2}
\end{eqnarray}%
}
\end{theorem}

\begin{proof}
It is a consequence of Conclusion \ref{concl1} and Theorems \ref{th2a} and %
\ref{th2b} for superpositions of ansatzes (\ref{class1}) and (\ref{class2})
resulting into (\ref{ans1}). If we repeat the computations from Appendix \ref%
{sb} for geometric data (\ref{paramdcoef}) considering that contributions of
type $\widehat{\Gamma }\cdot \widehat{\Gamma }$ are small, we can see that
the equations (\ref{eq2})--(\ref{eq4}) are derived to be, respectively,
equivalent to sums of (\ref{eq2b})--(\ref{eq4b}) and (\ref{eq2c})--(\ref%
{eq4c}). The torsionless conditions (\ref{lcconstr1}) consist a sum of
similar conditions (\ref{lccond}) and (\ref{lccondd}). $\square $
\end{proof}

\vskip5pt

In general, the solutions defined by a system (\ref{eq1})--(\ref{conf2}) can
not be transformed into solutions parameterized by an ansatz (\ref{class1})
and/or (\ref{class2}). As we shall prove in Section \ref{s4}, the general
solutions of the such systems of PDEs are determined by corresponding sets of
generating and integration functions. A solution for (\ref{eq1})--(\ref%
{conf2}) contains a larger set of $h-v$--generating functions than those
with some N--coefficients stated to be zero.

\section{Generic Off--Diagonal Solutions}

\label{s4} The goal of this section is to show how decoupling effective
Einstein equations we can integrate such PDE in very general forms depending
on the properties of the coefficients of an ansatz for metrics. Some similar theorems
for 4--d, 5--d and higher dimension modifications of GR were proven in Refs
\cite{vexsol1,vexsol2,ijgmmp}. In this work, we generalize those results for
$f$--modifications.

\subsection{Generating solutions with weak one Killing symmetry}

We prove that the modified gravitational equations encoding gravitational
and generic off--diagonal nonlinear interactions and satisfying the
conditions of Theorem \ref{th2a} can be integrated in general forms for $%
h_{a}^{\ast }\neq 0$ and certain special cases with zero and non--zero
sources (\ref{source}). In general, such generic off--diagonal metrics are
determined by generating functions depending on three/four coordinates. The
bulk of known exact solutions with diagonalizable metrics and coefficients
depending on two coordinates (in certain special frames of references) can
be included as special cases for more general nonholonomic configurations.

\subsubsection{(Non) vacuum metrics with $h_{a}^{\ast }\neq 0$}

For ansatz (\ref{ans1}) with data $\omega =1,\underline{h}_{3}=1,\underline{w%
}_{i}=0$ and $\underline{n}_{i}=0$ for (\ref{paramdcoef}), when $h_{a}^{\ast
}\neq 0,$ and the condition that the source
\begin{equation}
\mathbf{\Upsilon }_{~\delta }^{\beta }=diag[\mathbf{\Upsilon }_{\alpha }:%
\mathbf{\Upsilon }_{~1}^{1}=\mathbf{\Upsilon }_{~2}^{2}=\Upsilon
(x^{k},y^{3});\mathbf{\Upsilon }_{~3}^{3}=\mathbf{\Upsilon }%
_{~4}^{4}=~^{v}\Upsilon (x^{k})],  \label{source1a}
\end{equation}%
is not zero, the solutions of gravitational field equations can be
constructed as follows

\begin{theorem}
\label{th3a}The system (\ref{eq1b})--(\ref{eq4b}) with source (\ref%
{source1a}) can be integrated in general forms by
\begin{eqnarray}
\mathbf{g} &=&\epsilon _{i}e^{\psi (x^{k})}dx^{i}\otimes
dx^{i}+h_{3}(x^{k},y^{3})\mathbf{e}^{3}\otimes \mathbf{e}%
^{3}+h_{4}(x^{k},y^{3})\underline{h}_{4}(x^{k},y^{4})\mathbf{e}^{4}\otimes
\mathbf{e}^{4},  \notag \\
\mathbf{e}^{3} &=&dy^{3}+w_{i}(x^{k},y^{3})dx^{i},\ \mathbf{e}%
^{4}=dy^{4}+n_{i}(x^{k})dx^{i},  \label{sol1}
\end{eqnarray}%
with coefficients determined by generating functions $\psi (x^{k}),\phi
(x^{k},y^{3}),$ $\phi ^{\ast }\neq 0,$ $n_{i}(x^{k})$ and $\underline{h}%
_{4}(x^{k},y^{4}),$ and integration functions $~^{0}\phi (x^{k})$ following
recurrent formulas and conditions%
\begin{eqnarray}
\epsilon _{1}\psi ^{\bullet \bullet }+\epsilon _{2}\psi ^{\prime \prime }
&=&2\ \ ^{v}\Upsilon ;  \label{sol1a} \\
h_{4} &=&\pm \frac{1}{4}\int \left\vert \Upsilon \right\vert ^{-1}\left(
e^{2\phi }\right) ^{\ast }dy^{3},\mbox{ or \ }\   = \pm \frac{1}{4\Lambda }e^{2[\phi -~^{0}\phi ]},\mbox{ if }\Upsilon
=\Lambda =const;  \label{sol1b} \\
h_{3} &=&\pm \left[ \left( \sqrt{|h_{4}|}\right) ^{\ast }\right]
^{2}e^{-2\phi }=\frac{\phi ^{\ast }}{2\left\vert \Upsilon \right\vert }%
\left( \ln |h_{4}|\right) ^{\ast }\mbox{ or \ } \  \notag = \pm \frac{\left( \phi ^{\ast }\right) ^{2}}{4\Lambda },\mbox{ if
}\Upsilon =\Lambda =const;  \label{sol1c} \\
w_{i} &=&-\partial _{i}\phi /\phi ^{\ast },  \label{sol1d}
\end{eqnarray}%
where constraints
\begin{equation}
w_{i}^{\ast } = (\partial _{i}-w_{i}\partial _{3})\ln |h_{4}|,\partial
_{k}w_{i}=\partial _{i}w_{k},\   n_{k}\underline{h}_{4}^{\circ } = \mathbf{\partial }_{k}\underline{h}%
_{4},\partial _{i}n_{k}=\partial _{k}n_{i}.   \label{lccondm}
\end{equation}%
must be imposed in order to satisfy the zero torsion conditions (\ref{lccond}%
); we should take respective values $\epsilon _{i}=\pm 1$ and $\pm $ in (\ref%
{sol1b}) and/or (\ref{sol1c}) if we want to fix a necessary spacetime
signature.
\end{theorem}

\begin{proof}
We sketch a proof which transforms into similar ones in \cite%
{vexsol1,vexsol2,ijgmmp} if $\underline{h}_{4}=1.$

\begin{itemize}
\item A horizontal metric $g_{i}(x^{2})$ is for 2--d and can be always
represented in a conformally flat form $\epsilon _{i}e^{\psi
(x^{k})}dx^{i}\otimes dx^{i}.$ For such an h--metric, the equation (\ref{eq1b}%
) is a 2-d Laplace/wave equation (\ref{sol1a}) which can be solved exactly
if a source $^{v}\Upsilon (x^{k})-4s^{2}$ is prescribed from certain
physical conditions.

\item If $h_{4}^{\ast }\neq 0,$ we can define nontrivial functions
\begin{equation}
\phi =\ln \left\vert \frac{h_{4}^{\ast }}{\sqrt{|h_{3}h_{4}|}}\right\vert ,\
\gamma :=\left( \ln \frac{|h_{4}|^{3/2}}{|h_{3}|}\right) ^{\ast },\ \ \alpha
_{i}=h_{4}^{\ast }\partial _{i}\phi ,\ \beta =h_{4}^{\ast }\phi ^{\ast }
\label{auxfunct}
\end{equation}%
for a function $\phi (x^{k},y^{3}).$ If $\phi ^{\ast }\neq 0,$ we can write,
respectively, the equations (\ref{eq2b})--(\ref{eq4b}) in the forms,
\begin{eqnarray}
\phi ^{\ast }h_{4}^{\ast } &=&2h_{3}h_{4}\Upsilon .\   \label{eq2bb} \\
\beta w_{i}+\alpha _{i} &=&0,\   \label{eq3bb} \\
n_{i}^{\ast \ast }+\gamma n_{i}^{\ast } &=&0.  \notag
\end{eqnarray}%
For the last equation, we must take any trivial solution given by functions $%
n_{i}(x^{k})$ satisfying the conditions $\partial _{i}n_{j}=\partial
_{j}n_{i}$ in order to solve the constraints (\ref{lccond}). Using
coefficients (\ref{auxfunct}) with $\alpha _{i}\neq 0$ and $\beta \neq 0,$
we can always express $w_{i}$ via derivatives of $\phi ,$ i.e., in the form (%
\ref{sol1d}). We can choose any $\phi (x^{k},y^{3})$ with $\phi ^{\ast }\neq
0 $ as a generating function and express $h_{4}$ and, after that, $h_{3}$ as
some integrals/derivatives of functions depending on $\phi $ and source $%
\Upsilon (x^{k},y^{3}),$ see corresponding formulas (\ref{sol1b}) and (\ref%
{sol1c}). The integrals can be computed in a general explicit form if $\
\Upsilon (x^{k},y^{3})=\Lambda =const,$ when possible matter field
interactions and $f$--modifications are approximated by an energy--momentum
tensor as a cosmological constant,
\begin{equation}
\mathbf{\Upsilon }_{~\delta }^{\beta }=\delta _{~\delta }^{\beta }\Upsilon
=\delta _{~\delta }^{\beta }\Lambda .  \label{cosmconst}
\end{equation}

\item A possible dependence on $y^{4}$ is present in function $\underline{h}%
_{4}$ which must satisfy conditions of type (\ref{aux4}) in order to be
compatible with (\ref{lccond}). It is not possible to write, in explicit form,
the solutions for the zero torsion condition if the source $\mathbf{\Upsilon
}_{~\delta }^{\beta }$ is parameterized by arbitrary functions. Nevertheless,
if $\mathbf{\Upsilon }_{~\delta }^{\beta }$ is of type (\ref{cosmconst}), we
get $h_{4}\sim e^{2\phi }$ and $w_{i}\sim \partial _{i}\phi /\phi ^{\ast }$
positively solve the constraints
\begin{equation}
w_{i}^{\ast }=\left( \partial _{i}-w_{i}\partial _{3}\right) \ln |h_{4}|%
\mbox{ and }\partial _{i}w_{j}=\partial _{j}w_{i},  \label{aux4b}
\end{equation}%
transformed into $\phi ^{\ast \ast }\partial _{i}\phi -\phi ^{\ast
}(\partial _{i}\phi )^{\ast }=0$. By straightforward computation, we can
check that (\ref{lccond}) are satisfied by (\ref{lccondm}) when $n_{i}^{\ast
}=0$ and $w_{i}$ is determined by (\ref{sol1d}). $\square $
\end{itemize}
\end{proof}

\vskip5pt

The solutions constructed in Theorem \ref{th3a}, and those which can be
derived following Corollary \ref{corol1} are very general ones and contain
as particular cases (perhaps) all known exact solutions for (non) holonomic
effective Einstein spaces with Killing symmetries. They also can be
generalized to include arbitrary finite sets of parameters as it is proven
in Ref. \cite{ijgmmp}.

\begin{remark}
\label{remex}\textbf{(--Example) }\ Introducing $H=-\ln |h_{4}|,$ the system
of equations (\ref{eq2bb}), (\ref{eq3bb}) and (\ref{aux4b}) can be written
in the form
\begin{eqnarray}
h_{3} &=&-\phi ^{\ast }~H^{\ast }/2\Upsilon _{2},  \label{r1} \\
\phi ^{\ast }w_{i}+\partial _{i}\phi &=&0,  \label{r2} \\
w_{i}^{\ast }+\partial _{i}H-w_{i}H^{\ast } &=&0.  \label{r3}
\end{eqnarray}%
Considering $H(x^{k},y^{3})$ as a generating function, we can integrate (\ref%
{r3}) in explicit form using parameterizations $w_{i}= _{1}w_{i}(x^{k},y^{3})
\times$ $_{2}w_{i}(x^{k},y^{3})$ (in this formula, we do not consider
summation on $i$). We choose $\phi (x^{k},y^{3})$ as a nontrivial solution of
the system of first order PDEs (\ref{r2}) for any found $w_{i}.$ So, we can
define $h_{3}$ for any prescribed source $\Upsilon _{2}$ in (\ref{r1}). We
conclude that the LC--conditions (\ref{lccond}) (relating $w_{i}$ to $h_{4},$
which does not allow a "complete" decoupling) can be solved also in very
general forms by further fixing of parameterizations, classes of functions
and boundary conditions for $H,\phi $ and $\Upsilon _{2}.$
\end{remark}

Nevertheless, we note that an explicit example of exact solutions can be
derived by taking derivative on $\ast $ for (\ref{r2}) and using a subclass
of functions $\phi $ when $\left( \partial _{i}\phi \right) ^{\ast
}=\partial _{i}(\phi ^{\ast }).$ We obtain $w_{i}\frac{\phi ^{\ast \ast }}{%
\phi ^{\ast }}+w_{i}^{\ast } + \frac{\partial _i \phi ^{\ast }}{\phi ^{\ast }%
}=0$. This system is equivalent/compatible to (\ref{r3}) if, for instance,
 $\Phi :=\ln |\phi ^{\ast }|,$ $H^{\ast }=-\Phi ^{\ast },~\partial
_{i}H=\partial _{i}\Phi$. Such equations are satisfied by any functions $%
H(~^{-}\varsigma )$ and $\Phi (~^{+}\varsigma ),$ where $~^{\pm }\varsigma
:=\pm v-x^{1}-x^{2}$.

\subsubsection{Effective modified vacuum gravitational configurations}

We can consider a subclass of generic off--diagonal modified gravitation
interactions which can be encoded as effective Einstein manifolds. In
general, such classes of solutions do not have a smooth limit from
non-vacuum to vacuum models.

\begin{corollary}
\label{corolvacuum}The effective vacuum solutions for the modified
gravitational equations with ansatz for metrics of type (\ref{sol1}) with
vanishing source (\ref{source1a}) are paramet\-riz\-ed in the form
\begin{eqnarray}
\mathbf{g} &=&\epsilon _{i}e^{\psi (x^{k})}dx^{i}\otimes
dx^{i}+h_{3}(x^{k},y^{3})\mathbf{e}^{3}\otimes \mathbf{e}%
^{3}+h_{4}(x^{k},y^{3})\underline{h}_{4}(x^{k},y^{4})\mathbf{e}^{4}\otimes
\mathbf{e}^{4},  \notag \\
\mathbf{e}^{3} &=&dy^{3}+w_{i}(x^{k},y^{3})dx^{i},\ \mathbf{e}%
^{4}=dy^{4}+n_{i}(x^{k})dx^{i},  \label{sol2}
\end{eqnarray}%
where coefficients are defined by solutions of the system
\begin{eqnarray}
\ddot{\psi}+\psi ^{\prime \prime } &=&0,  \label{ep1a} \\
\phi ^{\ast }\ h_{4}^{\ast } &=&0,  \label{ep2a} \\
\beta w_{i}+\alpha _{i} &=&0,  \label{ep3a}
\end{eqnarray}%
and with coefficients  computed following formulas (\ref{auxfunct}) for
nonzero $\phi ^{\ast }\ $\ and $h_{4}^{\ast }$ with possible further zero
limit; such coefficients and $h_{3}$ and $\underline{h}_{4}$ are subjected,
additionally, to the zero--torsion conditions (\ref{lccondm}).
\end{corollary}

\begin{proof}
Considering the system of equations (\ref{eq1b})-- (\ref{eq4b}) with zero
right sides, we obtain, respectively, the equations (\ref{ep1a})-- (\ref{ep3a}%
). \ For positive signatures on h--subspace and equation (\ref{ep1a}), we
can take $\psi =0,$ or consider a trivial 2-d wave equation if one of
coordinates $x^{k}$ is timelike. There are two possibilities to satisfy the
condition (\ref{ep2a}). The first one is to consider that $%
h_{4}=h_{4}(x^{k}),$ i.e., $h_{4}^{\ast }$ which states that the equation (%
\ref{ep2a}) has solutions with zero source for arbitrary function $%
h_{3}(x^{k},y^{3})$ and arbitrary N--coefficients $w_{i}(x^{k},y^{3})$ as
follows from (\ref{auxfunct}). For such vacuum configurations, the functions
$h_{3}$ and $w_{i}$ can be taken as generation ones which should be
constrained only by the conditions (\ref{lccondm}). Equations of type (\ref%
{aux4b}) constrain substantially the class of admissible $w_{i}$ if $h_{4}$
depends only on $x^{k}.$ Nevertheless, $h_{3}$ can be an arbitrary one
generating solutions which can be extended for nontrivial sources $%
\underline{\Upsilon }$ and systems (\ref{eq1c})-- (\ref{lccondd}) and/or (%
\ref{eq1})-- (\ref{conf2}).

A different class of solutions can be generated if we state, after
corresponding coordinate transforms, $\phi =\ln \left| h_{4}^{\ast }/\sqrt{%
|h_{3}h_{4}|}\right| =\ ^{0}\phi =const,~\phi ^{\ast }=0$. For such
configurations, we can consider $h_{4}^{\ast }\neq 0,$ and solve (\ref{ep2a}%
) as
\begin{equation}
\sqrt{|h_{3}|}=\ ^{0}h(\sqrt{|h_{4}|})^{\ast },  \label{rel1}
\end{equation}%
for $\ ^{0}h=const\neq 0.$ Such v--metrics are generated by any $%
f(x^{i},y^{3}),\ f^{\ast }\neq 0,$ when
\begin{equation}
h_{4}=-f^{2}\left( x^{i},y^{3}\right) \mbox{ and }h_{3}=(\ ^{0}h)^{2}\ \left[
f^{\ast }\left( x^{i},y^{3}\right) \right] ^{2},  \label{aux2}
\end{equation}%
where the signs are such way fixed that for $N_{i}^{a}\rightarrow 0$ we
obtain diagonal metrics with signature $(+,+,+,-).$ The coefficients $\alpha
_{i}=\beta =0$ in (\ref{ep3a}) and $w_{i}(x^{k},y^{3})$ can be any functions
solving (\ref{lccondm}). This is equivalent to
\begin{eqnarray}
w_{i}^{\ast } &=&2\mathbf{\partial }_{i}\ln |f|-2w_{i}(\ln |f|)^{\ast },
\label{cond1} \\
\partial _{k}w_{i}-\mathbf{\partial }_{i}w_{k} &=&2(w_{k}\partial
_{i}-w_{i}\partial _{k})\ln |f|,  \notag
\end{eqnarray}%
for any $n_{i}(x^{k})$ when $\partial _{i}n_{k}=\partial _{k}n_{i}$.
Constraints of type $n_{k}\underline{h}_{4}^{\circ }=\mathbf{\partial }_{k}%
\underline{h}_{4}$ (\ref{aux4}) have to be imposed for a nontrivial multiple
$\underline{h}_{4}.$ $\square $
\end{proof}

\vskip5pt

Using Corollary \ref{corol1}, the ansatz (\ref{sol2}) can be "dualized" to
generate effective vacuum solutions with weak Killing symmetry on $\partial
/\partial y^{3}.$ Finally, we note that the signature of the generic
off--diagonal metrics generated in this subsection depend on the fact which
coordinate $x^1,x^2,y^3 $ or $y^4$ is chosen to be a timelike one.

\subsection{Non--Killing effective Einstein configurations}

The Theorem \ref{th2b} can be applied for constructing non--vacuum and
effective vacuum solutions of the modified gravitational equations depending
on all coordinates without explicit Killing symmetries.

\subsubsection{Non--vacuum off--diagonal solutions}

We can generate such effective Einstein manifolds as follows
\begin{corollary}
\label{corym1}An ansatz of type (\ref{class1}) with d--metric {\small
\begin{eqnarray*}
\mathbf{g} &=&\epsilon _{i}e^{\psi }dx^{i}\otimes dx^{i}+\omega ^{2}[\pm
\frac{\left( \phi ^{\ast }\right) ^{2}}{4\Lambda }\mathbf{e}^{3}\otimes
\mathbf{e}^{3}\pm \frac{1}{4\Lambda }e^{2[\phi -~^{0}\phi ]}\underline{h}_{4}%
\mathbf{e}^{4}\otimes \mathbf{e}^{4}], \\
\mathbf{e}^{3} &=&dy^{3}-(\partial _{i}\phi /\phi ^{\ast })dx^{i},\ \mathbf{e%
}^{4}=dy^{4}+n_{i}(x^{k})dx^{i},
\end{eqnarray*}%
} where the coefficients are subjected to conditions (\ref{sol1a})--(\ref%
{lccondm}) and
 $\partial _{k}\omega +(\partial _{i}\phi /\phi ^{\ast })\omega ^{\ast
}-n_{i}\omega ^{\circ }=0$
defines solutions of the Einstein equations $R_{\alpha \beta }=\Lambda
g_{\alpha \beta }$ with nonholonomic interactions and modifications encoded
effectively into the vacuum structure of GR with nontrivial cosmological
constant, $\Lambda \neq 0.$
\end{corollary}

\begin{proof}
It is a consequence of Theorem \ref{th2b} and Corolarry \ref{th3a}. $%
\square $
\end{proof}

\vskip5pt

In a similar form, we can generate solutions of type (\ref{class2}) when the
conformal factor is a solution of $\partial _{k}\omega -\underline{w}%
_{i}(x^{k})\omega ^{\ast }+(\partial _{i}\underline{\phi }/\underline{\phi }%
^{\circ })\omega ^{\circ }=0$ with respective "dual" generating functions $%
\omega $ and $\underline{\phi }$ when the data (\ref{sol1a})--(\ref{lccondm}%
) $\ $are re--defined for solutions with weak Killing symmetry on $\partial
/\partial y^{3}.$

\subsubsection{Effective vacuum off--diagonal solutions}

Vacuum Einstein spaces encoding nonholonomic interactions and $f$%
--modifica\-tions can be constructed using

\begin{corollary}
\label{corym2}An ansatz of type (\ref{class1}) with d--metric
{\small
\begin{eqnarray*}
\mathbf{g} &=&\epsilon _{i}e^{\psi (x^{k})}dx^{i}\otimes dx^{i}+\omega
^{2}(x^{k},y^{a})[(\ ^{0}h)^{2}\ \left[ f^{\ast }\left( x^{i},y^{3}\right) %
\right] ^{2}\mathbf{e}^{3}\otimes \mathbf{e}^{3} -f^{2}\left( x^{i},y^{3}\right) \underline{h}_{4}(x^{k},y^{4})\mathbf{e}%
^{4}\otimes \mathbf{e}^{4}], \\
\mathbf{e}^{3} &=&dy^{3}+w_{i}(x^{k},y^{3})dx^{i},\ \mathbf{e}%
^{4}=dy^{4}+n_{i}(x^{k})dx^{i},
\end{eqnarray*}%
}
where the coefficients are subjected to conditions (\ref{rel1})--(\ref{cond1}%
), (\ref{lccondm}) and $\partial _{k}\omega -w_{i}\omega ^{\ast
}-n_{i}\omega ^{\circ }=0$ defines generic off--diagonal solutions of $%
R_{\alpha \beta }=0$.
\end{corollary}

\begin{proof}
It is a consequence of Theorem \ref{th2b} and Corollary \ref{corolvacuum}. $%
\square $
\end{proof}

\vskip5pt

Solutions of type (\ref{class2}) can be defined if the conformal factor is a
solution of $\partial _{k}\omega -\underline{w}_{i}(x^{k})\omega ^{\ast }-%
\underline{n}_{i}(x^{k},y^{4})\omega ^{\circ }=0$ with respective "dual"
generating functions $\omega (x^{k},y^{a})$ and $\underline{\phi }%
(x^{k},y^{4})$ when the data (\ref{rel1})--(\ref{cond1}) and (\ref{lccondm})
are re--defined for ansatz with weak Killing symmetry on $\partial /\partial
y^{3}.$

Summarizing the results of this section, we state

\begin{claim}
All generic off--diagonal non--vacuum and effective vacuum solutions of the
modified gravitational equations determined, respectively, by Theorem \ref%
{th3a} and Corollaries \ref{corolvacuum}, \ref{corym1} and \ref{corym2} can
be generalized to metrics of type (\ref{ans1}). For such constructions,
there are used nonlinear superposition of metrics and their
 "duals"  on v--coordinates in order to
define non--Killing solutions of respective systems of PDEs from Theorems %
\ref{th2b} and/or \ref{th2}
\end{claim}

The statements of this Claim are formulated following our experience on
constructing generic off--diagonal vacuum and non--vacuum solutions. For
such nonlinear systems, it is not possible to formulate certain general
uniqueness and exhaustive criteria. Sure, not all solutions of modified
gravitational equations can be constructed in such forms, or related to any
sets of prescribed solutions via "nondegernerate" nonholonomic deformations.
The length of this paper does not allow us to present all technical details
and general formulas for coefficients for ansatzes for d--metrics\footnote{%
Such formulas are, for instance, of type (\ref{sol1a})--(\ref{lccondm}),
with functions $h_{3}(.,y^{3})$ and $h_{4}(.,y^{3})$ and further
 "dualization"  $h_{3}\rightarrow
\underline{h}_{4}(.,y^{4})$ and $h_{4}\rightarrow \underline{h}_{3}(.,y^{4})$
with corresponding re--definition of N--connection coefficients.} which can
be constructed following this Claim. Explicit examples supporting our
approach are given in sections \ref{s5} and \ref{s6}.

\section{The Cauchy Problem and Decoupling}

\label{s5} The gravitational interactions in modified gravity studied in
this work are described by off--diagonal solutions of
\begin{equation}
R_{\alpha \beta }=\Lambda g_{\alpha \beta },  \label{neinstm}
\end{equation}%
which can be found in very general forms with respect to N--adapted frames
for certain nonintegrable spacetime $2+2$ splitting of type $\mathbf{N}$ (%
\ref{ncon}). An effective cosmological constant $\Lambda $ encodes a
gravitational "vacuum" cosmological
constant $\Lambda $ in GR and $f$--modifications. Pseudo--Riemanian
manifolds with metrics $g_{\alpha \beta }$ adapted to chosen non--integrable
distribution with $2+2$ splitting and satisfying (\ref{neinstm}) are called
\textit{nonholonomic effective Einstein manifolds}. Hereafter, we shall refer
to such systems of PDEs as \textit{nonholonomic effective vacuum spacetimes},
regardless of whether, or not, an (effective) cosmological constant vanishes or
can be "polarized" by gravitational
and/or matter field interactions into some N--adapted diagonal sources
admitting formal integration of gravitational field equations.

The equations (\ref{neinstm}), and their N--adapted equivalents (\ref%
{cdeinst})--(\ref{lcconstr}), constitute a second--order system
quasi--linear PDE for the coefficients of spacetime metric $\mathbf{g}%
=\{g_{\alpha \beta }\}.$ This means that given a manifold $\mathbf{V}$ of
necessary smooth class such a quasi--linear system is linear in the second
derivatives of the metric and quadratic in the first derivatives $\partial
_{\gamma }g_{\alpha \beta }$ (the coefficients of such PDEs are rational
functions of $g_{\alpha \beta }).$ To be able to decouple and formally
integrate such systems is necessary to consider special classes of
nonholonomic frames and constraints. For GR, this type of equations do not
fall in any of the standard cases of hyperbolic, parabolic, or elliptic
systems which typically lead to unique solutions. It is important to
formulate criteria when such general solutions would be unique ones with a
topology and differential structure determined by some initial data. How the
diffeomorphism, or coordinate, invariance and arbitrary frame transforms
(principle of relativity) would be taken into account for N--splitting?

In mathematical relativity \cite{bruhat,klainerman,chrusch1,christod,chru2},
it was proven a fundamental result (due to Choquet--Bruhat, 1952) that there
exists a set of hyperbolic equations underlying (\ref{neinstm}). The goal of
this section is to study the evolution (Cauchy) problem for the system (\ref%
{cdeinst})--(\ref{lcconstr}) in N--adapted form and preserving the
decoupling property.

\subsection{The local N--adapted evolution problem}

In the evolutionary approach, the topology of spacetime manifold is chosen
in the form $\mathbf{V}=\mathbb{R}\times ~^{3}V,$ where $~^{3}V$ is a 3--d
manifold carrying initial data. It should be noted here that there is no a
priori known natural time--coordinate even we may fix a signature for metric
and chose certain coordinates to be time or space like ones.

\begin{definition}
A set of coordinates $\{\widehat{u}^{\mu }=(\widehat{x}^{i},\widehat{y}%
^{a})\}$ is canonically N--har\-mon\-ic, i.e., it is both harmonic and adapted
to a splitting $\mathbf{N}$ (\ref{ncon}), if each of the functions $\widehat{%
u}^{\mu }$ satisfies the wave equation%
\begin{equation}
\widehat{\square }\widehat{u}^{\mu }=0,  \label{nwc}
\end{equation}%
where the canonical d'Alembert operator $\widehat{\square }:=\widehat{%
\mathbf{D}}_{\alpha }$ $\widehat{\mathbf{D}}^{\alpha }$ acts on a scalar $%
f(x,y)$ in the form%
{\small
$$\widehat{\square }f :=(\sqrt{|\mathbf{g}_{\alpha \beta }|})^{-1}\mathbf{e}%
_{\mu }\left( \sqrt{|\mathbf{g}_{\alpha \beta }|}\mathbf{g}^{\mu \nu }%
\mathbf{e}_{\nu }f\right) =(\sqrt{|g_{kl}|})^{-1}\mathbf{e}_{i}\left( \sqrt{|g_{kl}|}g^{ij}\mathbf{e}%
_{j}f\right) +(\sqrt{|g_{cd}|})^{-1}\partial _{a}\left( \sqrt{|g_{kl}|}%
g^{ab}e_{b}f\right),$$
}
for a d--metric $\mathbf{g}_{\alpha \beta }=(g_{ij},g_{ab})$ (\ref{dm})
defined with respect to N--adapted (co) frames (\ref{nader})--(\ref{nadif})
and canonical d--connection $\widehat{\mathbf{D}}_{\alpha }$ (\ref{cdc}).
\end{definition}

We can say that such coordinates $\widehat{u}^{\mu }=(\widehat{x}^{i},%
\widehat{y}^{a})$ are N--adapted wave--coordinates.

\begin{lemma}
In canonical N--harmonic coordinates, the effective Einstein equations (\ref%
{neinstm}) re--defined in canonical d--connection variables (\ref{cdeinst})
can be written, equivalently,%
\begin{equation}
\widehat{\mathbf{E}}^{\alpha \beta } =\widehat{\square }\mathbf{g}^{\alpha
\beta }-\mathbf{g}^{\tau \theta }\left[ (\mathbf{g}^{\alpha \mu }\widehat{%
\mathbf{\Gamma }}_{~\mu \nu }^{\beta }+\mathbf{g}^{\alpha \mu }\widehat{%
\mathbf{\Gamma }}_{~\mu \nu }^{\beta })\widehat{\mathbf{\Gamma }}_{~\tau
\theta }^{\nu }+2\mathbf{g}^{\gamma \mu }\widehat{\mathbf{\Gamma }}_{~\mu
\theta }^{\alpha }\widehat{\mathbf{\Gamma }}_{~\tau \gamma }^{\beta }\right]
-2\Lambda \mathbf{g}^{\alpha \beta }=0,  \label{cdeinst1}
\end{equation}%
i.e., such PDEs, for $\mathbf{g}^{\alpha \beta }$ (using algebraic transforms,
for $\mathbf{g}_{\alpha \beta })$ form a system of second--order
quasi--linear N--adapted wave--type equations.
\end{lemma}

\begin{proof}
It is a standard computation with respect to N--adapted frames by using
formulas (\ref{cdc}), (\ref{enstdt}) and (\ref{driccic}). If the zero
torsion conditions (\ref{lcconstr}) are imposed, we get well known results
from GR but (in our case) adapted to $2+2$ non--integrable splitting.\ $%
\square $
\end{proof}

\vskip5pt

This Lemma allows us to apply the standard theory of hyperbolic PDEs (see,
for instance, \cite{evans}). Let us denote by $H_{loc}^{k}$ the Sobolev
spaces of functions which are in $L^{2}(K)$ for any compact set $K$ when
their distributional derivatives are considered up to an integer order $k$
also in $L^{2}(K).$ We shall also use N--adapted wave coordinates with
additional formal $3+1$ splitting, for instance, in a form $\widehat{u}%
^{\mu}=(\ ^{t}\widehat{u},\widehat{u}^{\overline{i}}),$ where $\ ^{t}%
\widehat{u}$ is used for the timelike coordinate and $\widehat{u}^{\overline{%
i}}$ are for 3 spacelike coordinates. "Hats" can be eliminated if such a
splitting is considered for arbitrary local coordinates. Standard results
from the theory of PDEs give rise to this

\begin{theorem}
\label{thinid}The field equations (\ref{cdeinst}) for nonholonomic effective
Einstein manifolds have a unique solution $\mathbf{g}^{\alpha \beta }$
defined by PDEs (\ref{cdeinst1}) stated on an open neighborhood $\mathcal{U}%
\subset \mathbb{R}\times \mathbb{R}^{3}$ of $\mathcal{O}\subset \{0\}\times
\mathbb{R}^{3}$ \ with any initial data%
\begin{equation}
\ \mathbf{g}^{\alpha \beta }(0,\widehat{u}^{\overline{i}})\in H_{loc}^{k+1}%
\mbox{ \  and \ }\ \frac{\partial \mathbf{g}^{\alpha \beta }(0,\widehat{u}^{%
\overline{i}})}{\partial (\ ^{t}\widehat{u})}\in H_{loc}^{k+1},k>3/2.
\label{inid}
\end{equation}%
The set $\mathcal{U}$ \ can be chosen in a form that $\left( \mathcal{U},%
\mathbf{g}^{\alpha \beta }\right) $ is globally hyperbolic with Cauchy
surface $\mathcal{O}$\footnote{%
Choosing corresponding classes of nonholonomic distributions $\mathbf{N}$ (%
\ref{ncon}), we can relax the conditions of differentiability as in Refs.
\cite{klainrodn,smtat} (we omit such constructions in this work).}.
\end{theorem}

There is no reason that a solution constructed using the anholonomic
deformation method as we considered in section \ref{s3} will satisfy the
wave conditions (\ref{nwc}), even if we  give certain initial data for
equation (\ref{cdeinst1}). In order to establish an hyperbolic (and
evolutionary) form of  such nonholonomic vacuum gravitational equations
we should re--define the Choquet--Bruhat problem (see details and references
in \cite{bruhat,klainerman,chrusch1,rodnianski,chru2}) with respect to
N--adapted frames. \ Using a 3+1 splitting of N--adapted coordinates, we
wirte $\mathbf{g}^{\alpha \beta }=(\mathbf{g}^{t\beta },\mathbf{g}^{%
\overline{i}\beta })$ and $\mathbf{e}_{\alpha }=(\mathbf{e}_{t},\mathbf{e}_{%
\overline{i}})$ and consider the d--vector field $n^{\mu }(x,y)$ of unit
timelike normals to the hypersurface $\{~^{t}\widehat{u}=0\}$\footnote{%
We do not use labels for coordinates like $0,1,2,3$ because the decoupling
property of the Einstein equations and general solutions can be proven for
arbitrary signature, for instance $\left( -+++\right) ,\left( ++-+\right)$,
etc.\  and for any set of local coordinates $u^{\alpha }$ with $\alpha
=1,2,3,4. $}. There is no loss of generality if we assume that on such a
hypersurface we have
\begin{equation}
\mathbf{g}^{tt}=-1\mbox{ and }\mathbf{g}^{t\overline{i}}=0.  \label{inic1}
\end{equation}%
We can state such conditions via additional N--adapted frame/coordinate
transform for any d--metric (\ref{dm}) with prescribed signature. It is also
possible to re--define the generating functions for any class of
off--diagonal solutions with decoupling of (modified) Einstein equations in
order to satisfy (\ref{inic1}).

Another necessary condition for the vanishing of a $\widehat{\square }u^{\mu
}$ is that this value is stated zero at $~^{t}\widehat{u}=0.$ This allows us
to express in N--adapted form the time like derivatives through N--elongated
space type ones,%
\begin{equation}
\mathbf{e}_{t}\left( \sqrt{|\mathbf{g}_{\alpha \beta }|}\mathbf{g}^{t%
\overline{i}}\right) =-\mathbf{e}_{\overline{i}}\left( \sqrt{|\mathbf{g}%
_{\alpha \beta }|}\mathbf{g}^{t\overline{i}}\right)  \label{inic2}
\end{equation}%
with N--elongated operators. So, the initial data from Theorem \ref{thinid}
can not be fixed in arbitrary form if we want to establish a hyperbolic
(evolutionary) character for nonholonomic vacuum Einstein equations, i.e. to
satisfy both systems (\ref{cdeinst1}) and (\ref{nwc}).  Really, the last
system of fist order PDEs allows us to compute the time--like derivatives $%
\partial \mathbf{g}^{t\overline{i}}(0,\widehat{u}^{\overline{i}})/\partial
~^{t}\widehat{u}\mid _{\{~^{t}\widehat{u}=0\}}$if $\mathbf{g}_{\overline{i}%
\overline{j}}\mid _{\{~^{t}\widehat{u}=0\}}$ and $\partial \mathbf{g}_{%
\overline{i}\overline{j}}/\partial ~^{t}\widehat{u}\mid _{\{~^{t}\widehat{u}%
=0\}}$have been defined. We conclude that the essential data for formulating
a N--adapted evolution problem should be formulated for a 3--d space
d--metric%
\begin{equation}
~^{[3]}\mathbf{g:=g}_{\overline{i}\overline{j}}(x,y)\mathbf{e}^{\overline{i}%
}\otimes \mathbf{e}^{\overline{j}},  \label{dm3d}
\end{equation}%
where $\mathbf{e}^{\overline{j}}$ are N--elongated differentials of type (%
\ref{nadif}), and its N--elongated (\ref{nader}) time--like derivatives.

Using the Theorem \ref{thinid} and above presented considerations, we get
the proof of the following theorem

\begin{theorem}
If the initial data (\ref{inid}) satisfy the conditions (\ref{inic1}) and (%
\ref{inic2}) and the so--called effective Einstein N--adapted constraint
equations,%
\begin{equation}
\left( \widehat{\mathbf{E}}_{\alpha \beta }+\Lambda \mathbf{g}_{\alpha \beta
}\right) n^{\alpha }=0,  \label{vacuumconstr}
\end{equation}%
are computed for zero distortion in (\ref{distrel}), then the d--metric
stated by Theorem \ref{thinid} defines solutions of the nonholonomic vacuum
equations (\ref{neinstm}) and/or (\ref{cdeinst})--(\ref{lcconstr}).
\end{theorem}

This theorem gives us the possibility to state the Cauchy data for decoupled
modified gravity systems and their generic off--diagonal solutions in order
to generate N--adapted evolutions.

\subsection{On initial data sets and global nonholonomic evolution}

We adopt this convention for spacetime nonholonomic manifolds $\mathbf{V}=%
\mathbb{R}\times ~^{3}V,$ were $~^{3}V$ is a N--adapted 3-d manifold, when a
Whitney sum $\ T\ ^{3}V=h\ ^{3}V\oplus v\ ^{3}V$ is stated by a spacelike
nonholonomic distribution $~^{3}\mathbf{N}$ of type (\ref{ncon}) and there is
an embedding $\ e:\ ^{3}V\rightarrow \mathbf{V}$.

Using the d--metric $~^{[3]}\mathbf{g}$ (\ref{dm3d}), we can define the
second fundamental form $\mathbf{K}$ of a spacelike hypersurface $\ ^{3}V$
in $\mathbf{V,}$ $\widehat{\mathbf{K}}\mathbf{(X,Y):=g}(\widehat{\mathbf{D}}%
_{\mathbf{X}}\mathbf{n},\mathbf{Y}), \forall \mathbf{X\in }T\ ^{3}V$. The
unity d--vector $\ \mathbf{n}=n^{\alpha }\mathbf{e}_{\alpha }=n_{\overline{i}%
}\mathbf{e}^{\overline{i}}=n_{t}\mathbf{e}^{t}=\left( \sqrt{|g^{tt}|}\right)
^{-1}\mathbf{e}^{t},$ with normalization $\mathbf{g(n,n)}=\mathbf{g}^{\alpha
\beta }n_{\alpha }n_{\beta }=\mathbf{g}^{tt}(n_{t})^{2}=-1,$ is time--future
and normal to $~^{3}V.$ The value $\widehat{\mathbf{K}}=\{\widehat{K}_{%
\overline{i}\overline{j}}\}$ is the extrinsic canonical curvature d--tensor
of $~^{3}V.$ Imposing the zero torsion conditions (\ref{lcconstr}), $%
\widehat{\mathbf{K}}\rightarrow \{K_{\overline{i}\overline{j}}\},$ where $K_{%
\overline{i}\overline{j}}=-\frac{1}{2}\mathbf{g}^{t\alpha }\left( \mathbf{e}%
_{\overline{j}}\mathbf{g}_{\alpha \overline{i}}\mathbf{+e}_{\overline{i}}%
\mathbf{g}_{\alpha \overline{j}}-\mathbf{e}_{\alpha }\mathbf{g}_{\overline{i}%
\overline{j}}\right) n_{t}$ are components computed in standard form using
the Levi--Civita connection, but with respect to N--adapted frames. $\ $We
can invert the last formula and write $\partial _{t}\mathbf{g}_{\overline{i}%
\overline{j}}=2(\mathbf{g}^{tt}n_{t})^{-1}K_{\overline{i}\overline{j}}+$
\{terms determined by $\mathbf{g}_{\alpha \beta }$ and their
space--derivatives\}. Such formulas show that $K_{\overline{i}\overline{j}}$
and $\partial _{t}\mathbf{g}_{\overline{i}\overline{j}}$ are geometric
counterparts on hypersurfaces $\{~^{t}\widehat{u}=0\}.$

For the canonical d--connection $\widehat{\mathbf{D}}_{\alpha }=\left(
\widehat{D}_{t},\widehat{\mathbf{D}}_{\overline{i}}\right) $ and d--metric $%
~^{[3]}\mathbf{g}$ induced on a spacelike hypersurface in a Lorentzian
nonholonomic manifold $\mathbf{V,}$ we can derive a N--adapted variant of
Gauss--Codazzi equations%
\begin{equation*}
\ ^{[3]}\widehat{R}_{~\overline{j}\overline{k}\overline{l}}^{\overline{i}} =%
\widehat{\mathbf{R}}_{~\overline{j}\overline{k}\overline{l}}^{\overline{i}}+%
\widehat{K}_{~\overline{l}}^{\overline{i}}\widehat{K}_{\overline{j}\overline{%
k}}-\widehat{K}_{~\overline{k}}^{\overline{i}}\widehat{K}_{\overline{j}%
\overline{l}},\ \widehat{\mathbf{D}}_{\overline{i}}\widehat{K}_{\overline{j}%
\overline{k}}-\widehat{\mathbf{D}}_{\overline{j}}\widehat{K}_{\overline{i}%
\overline{k}} =\widehat{\mathbf{R}}_{~\overline{j}\overline{k}\overline{l}}^{%
\overline{i}}n^{\overline{l}}.
\end{equation*}
In these formulas, $~^{[3]}\widehat{R}_{~\overline{j}\overline{k}\overline{l}%
}^{\overline{i}}$ is the canonical curvature d--tensor of $~^{[3]}\mathbf{g,~%
}\widehat{\mathbf{R}}_{~\overline{j}\overline{k}\overline{l}}^{\overline{i}}$
is computed following formulas (\ref{dcurv}) as the spacetime canonical
d--curvature tensor, $\mathbf{n}$ is the timelike normal to hypersurface $%
~^{3}V$ when the local N--adapted coordinate system is such way chosen that
d--vectors $\mathbf{e}_{\overline{j}}$ are tangent to $~^{3}V.$ Contracting
indices, introducing divergence operator $\widehat{d}iv$ determined by
N--elongated partial derivatives (\ref{nader}), trace operator $tr$ and
absolute differential $d,$ we derive from above equations and (\ref{cdeinst}%
) the following system of equations:
\begin{eqnarray}
\widehat{d}iv\widehat{K}-d(tr\widehat{K}) &=&8\pi \widehat{J},%
\mbox{
momentum constraint};\   \label{undetconstr} \\
\ _{[3]}^{s}\widehat{R}-2\Lambda -~^{[3]}\mathbf{g}^{\overline{l}\overline{j}%
}\widehat{K}_{~\overline{l}}^{\overline{i}}\widehat{K}_{\overline{j}%
\overline{i}}+(tr\widehat{K})^{2} &=&16\pi \widehat{\rho },%
\mbox{
Hamiltonian constraint};  \notag \\
\mathcal{C}(\widehat{\mathcal{F}},~^{[3]}\mathbf{g}) &=&0,%
\mbox{ Einstein
constraint eqs};  \notag
\end{eqnarray}%
where $\ _{[3]}^{s}\widehat{R}$ is computed as the scalar (\ref{sdcurv}) but
for $~^{[3]}\mathbf{g.}$ In a general context, we consider that $~^{3}V$ is
embedded in a nonholonomic spacetime $\mathbf{V}$ with induced data $\left(
~^{3}V,~^{[3]}\mathbf{g,}\widehat{\mathbf{K}},\widehat{\mathcal{F}}\right) ,$
we have $\widehat{J}:=-\left( \mathbf{n,\cdot }\right) $ and $\widehat{\rho }%
:=\left( \mathbf{n,n}\right) .$ The term $\mathcal{C}(\widehat{\mathcal{F}}%
,~^{[3]}\mathbf{g}) $ denotes the set of additional constraints resulting
from the non--gravitational part, including nonholonomic distributions $~%
\mathbf{N}$ (\ref{ncon}). If in such a set, there are included $~$the zero
torsion conditions (\ref{lcconstr}), we can omit "hats" on
geometric/physical objects if they are written in "not" N--adapted frames
of reference. The equations (\ref{undetconstr}) form an undetermined system
of PDEs. For 3-d, there are locally four equations for twelve unknown values
given by the components of d--tensors $~^{[3]}\mathbf{g} $ and $\widehat{%
\mathbf{K}}.$ Using the conformal method with the Levi--Civita connection,
see details and references in Ref. \cite{york}, we can study the existence
and uniqueness of solutions to such systems.

The above considerations motivate the following definition

\begin{definition}
A canonical vacuum initial N--adapted data set is defined by a triple $%
\left( ~^{3}V,~^{[3]}\mathbf{g,}\widehat{\mathbf{K}}\right) ,$ where $%
(~^{[3]}\mathbf{g,}\widehat{\mathbf{K}})$ are defined as a solution of (\ref%
{undetconstr}).
\end{definition}

If the conditions (\ref{lcconstr}) are imposed, the data $(~^{[3]}\mathbf{g,}%
\widehat{\mathbf{K}})$ are equivalent to similar ones $(~^{[3]}\mathbf{g,K})$
with $\mathbf{K}$ computed for the Levi--Civita connection. Covering $~^{3}V$
by coordinate neighborhoods $\mathcal{O}_{u}$ of points $u\in ~^{3}V,$ we
can use Theorem \ref{thinid} to construct globally hyperbolic N--adapted
developments $\left( \mathcal{U}_{u},\mathbf{g}_{u}\right) $ of an initial
data set $\left( \mathcal{U}_{u}\subset ~^{3}V,~^{[3]}\mathbf{g,}\widehat{%
\mathbf{K}}\right) $ as in the above definition. The d--metrics generated in
such forms will coincide after performing suitable N--adapted
frame/coordinated transforms on charts covering such a spacetime wherever
such solutions are defined. We can patch all data $\left( \mathcal{U}_{u},%
\mathbf{g}_{u}\right) $ together to a globally hyperbolic Lorentzian
nonholonomic spacetime containing a Cauchy surface $~^{3}V.$ We prove that

\begin{theorem}
Any N--adapted initial data $(\ ^{3}V,~^{[3]}\mathbf{g},\widehat{\mathbf{K}}%
) $ of differentiability class $H^{s+1}\times H^{s},s>3/2,$ admits a
globally hyperbolic, N--adapted and unique development (in the sense of
Theorem \ref{thinid} and of the above considered assumptions and proofs).
\end{theorem}

\section{Anholonomic Modifications of Black Holes}

\label{s6}

\subsection{(Non) holonomic non--Abelian effective vacuum spaces}

This class of effective vacuum solutions are generated not just as a simple
limit $\Lambda \rightarrow 0,$ for \ instance, for a class of solutions (\ref%
{sol1}) with coefficients (\ref{sol1a})--(\ref{sol1c}). We have to construct
off--diagonal solutions of the effective Einstein equations for the
canonical d--connection taking the vacuum equations $\widehat{\mathbf{R}}%
_{\alpha \beta }=0$ and an ansatz $~\mathbf{g}$  with coefficients
satisfying the conditions%
\begin{eqnarray}
&&\epsilon _{1}\psi ^{\bullet \bullet }(r,\theta )+\epsilon _{2}\psi
^{^{\prime \prime }}(r,\theta )=0;  \label{vacdsolc} \\
h_{3} &=&\pm e^{-2\ ^{0}\phi }\frac{\left( h_{4}^{\ast }\right) ^{2}}{h_{4}}%
\mbox{ for a given }h_{4}(r,\theta ,\varphi ),\ \phi (r,\theta ,\varphi )=\
^{0}\phi =const;\   \notag \\
w_{i} &=&w_{i}(r,\theta ,\varphi ),\mbox{ for any such functions if }\lambda
=0;  \notag \\
n_{i} &=&\ \left\{
\begin{array}{rcl}
\ \ \ ^{1}n_{i}(r,\theta )+\ ^{2}n_{i}(r,\theta )\int \left( h_{4}^{\ast
}\right) ^{2}|h_{4}|^{-5/2}dv,\ \  & \mbox{ if \ } & n_{i}^{\ast }\neq 0; \\
\ ^{1}n_{i}(r,\theta ),\quad \qquad \qquad \qquad \qquad \qquad &
\mbox{ if
\ } & n_{i}^{\ast }=0.%
\end{array}%
\right.  \notag
\end{eqnarray}
Vacuum solutions of the effective Einstein equations for the Levi--Civita
connection, i.e of $R_{\alpha \beta }=0,$ are generated if we impose
additional constraints on coefficients of d--metric, for $e^{-2\ ^{0}\phi
}=1,$ as solutions of (\ref{lccondm}),
\begin{eqnarray}
h_{3} &=&\pm 4\left[ \left( \sqrt{|h_{4}|}\right) ^{\ast }\right] ^{2},\quad
h_{4}^{\ast }\neq 0;  \label{auxvacsol} \\
w_{1}w_{2}\left( \ln |\frac{w_{1}}{w_{2}}|\right) ^{\ast } &=&w_{2}^{\bullet
}-w_{1}^{\prime },\ w_{i}^{\ast }\neq 0;\ w_{2}^{\bullet }-w_{1}^{\prime
}=0,\ w_{i}^{\ast }=0;  \notag \\
\ ^{1}n_{1}^{\prime }(r,\theta )-\ ^{1}n_{2}^{\bullet }(r,\theta ) &=&0,\
n_{i}^{\ast }=0.  \label{vaclcsoc}
\end{eqnarray}

The constructed class of vacuum solutions with coefficients subjected to
conditions (\ref{vacdsolc})--(\ref{vaclcsoc}) is of type (\ref{sol2}) for (%
\ref{ep1a})--(\ref{ep3a}). Such metrics consist a particular case of vacuum
ansatz defined by Corollary \ref{corym2} with $\underline{h}_{4}=1$ and $%
\omega =1.$

Here we note that former analytic and numeric programs (for instance,
standard ones with Maple/ Mathematica) for constructing solutions in gravity
theories can not be directly applied for alternative verifications of our
solutions. Those approaches do not encode the nonholonomic constraints which
we use for constructing integral varieties. Nevertheless, it is possible to
check, in general, the analytic form, see all details summarized in Refs. \cite%
{ijgmmp,vexsol1,vexsol2} (and formulas from Appendix \ref{sb}), that the
vacuum effective Einstein equations (\ref{eq1})--(\ref{lccond}) with zero
effective sources, $~^{v}\Upsilon =0~$ and $~\Upsilon =0,$ can be solved by using the
above presented off--diagonal ansatz for metrics.

\subsection{$f$--deformations of the Schwarzschild metric}

We can consider a "prime" metric which, in general, is not a solution of
(modified) Einstein equations,
\begin{equation}
~^{\varepsilon }\mathbf{g}=-d\xi \otimes d\xi -r^{2}(\xi )\ d\vartheta
\otimes d\vartheta -r^{2}(\xi )\sin ^{2}\vartheta \ d\varphi \otimes
d\varphi +\varpi ^{2}(\xi )\ dt\otimes \ dt.  \label{5aux1}
\end{equation}%
We shall deform nonholonomically this metric into a "target" off--diagonal
one which will be a solution of the vacuum Einstein equations. The
nontrivial metric coefficients in (\ref{5aux1}) are stated in the form
\begin{equation}
\check{g}_{1}=-1,\ \check{g}_{2}=-r^{2}(\xi ),\ \check{h}_{3}=-r^{2}(\xi
)\sin ^{2}\vartheta ,\ \check{h}_{4}=\varpi ^{2}(\xi ),  \label{5aux1p}
\end{equation}%
for local coordinates $x^{1}=\xi ,x^{2}=\vartheta ,y^{3}=\varphi ,y^{4}=t,$
where  $\xi =\int dr\ \left\vert 1-\frac{2\mu _{0}}{r}+\frac{\varepsilon }{r^{2}}%
\right\vert ^{1/2}$ and $\varpi ^{2}(r)=1-\frac{2\mu _{0}}{r}+\frac{%
\varepsilon }{r^{2}}$.
If we put $\varepsilon =0$ with $\mu _{0}$ considered as a point mass, the
metric $~^{\varepsilon }\mathbf{g}$ (\ref{5aux1}) determines the
Schwarzschild solution. For simplicity, we  analyze only the case of
"pure" gravitational vacuum solutions,
not considering a more general construction when $\varepsilon =e^{2}$ can be
related to the electric charge for the Reissner--Nordstr\"{o}m metric. In
our approach, $\varepsilon $ is a small parameter (eccentricity) defining a
small deformation of a circle into an ellipse.

We generate exact solutions of the system (\ref{ep1a})--(\ref{ep3a}) with
effective $\Lambda =0$ via nonholonomic deformations $\ ^{\varepsilon }%
\mathbf{g\rightarrow }\ _{\eta }^{\varepsilon }\mathbf{g}$, when $g_{i}=\eta
_{i}\check{g}_{i}$ and $h_{a}=\eta _{a}\check{h}_{a}$ and $w_{i},n_{i}$
define a target metric {\small
\begin{eqnarray}
~_{\eta }^{\varepsilon }\mathbf{g} &=&\eta _{1}(\xi )d\xi \otimes d\xi +\eta
_{2}(\xi )r^{2}(\xi )\ d\vartheta \otimes d\vartheta +  \eta _{3}(\xi ,\vartheta ,\varphi )r^{2}(\xi )\sin ^{2}\vartheta \ \delta
\varphi \otimes \delta \varphi -\eta _{4}(\xi ,\vartheta ,\varphi )\varpi
^{2}(\xi )\ \delta t\otimes \delta t, \notag  \\
\delta \varphi &=&d\varphi +w_{1}(\xi ,\vartheta ,\varphi )d\xi +w_{2}(\xi
,\vartheta ,\varphi )d\vartheta ,\ \delta t=dt+n_{1}(\xi ,\vartheta )d\xi
+n_{2}(\xi ,\vartheta )d\vartheta .  \label{5sol1}
\end{eqnarray}%
} The gravitational field equations for zero source relate the coefficients
of the vertical metric and polarization functions,
\begin{equation}
h_{3}=h_{0}^{2}(b^{\ast })^{2}=\eta _{3}(\xi ,\vartheta ,\varphi )r^{2}(\xi
)\sin ^{2}\vartheta ,\ h_{4}=-b^{2}=-\eta _{4}(\xi ,\vartheta ,\varphi
)\varpi ^{2}(\xi ),  \label{aux41}
\end{equation}%
for $|\eta _{3}|=(h_{0})^{2}|\check{h}_{4}/\check{h}_{3}|[(\sqrt{|\eta _{4}|}%
)^{\ast }]^{2}.$ In these formulas, we have to choose $h_{0}=const$\ ( $%
h_{0}=2$ in order to satisfy the first condition (\ref{vaclcsoc})), where $%
\check{h}_{a}$ are stated by the Schwarzschild solution for the chosen
system of coordinates and $\eta _{4}$ can be any function satisfying the
condition $\eta _{4}^{\ast }\neq 0.$ We generate a class of solutions for
any function $b(\xi ,\vartheta ,\varphi )$ with $b^{\ast }\neq 0.$ For
different purposes, it is more convenient to work directly with $\eta _{4},$
for $\eta _{4}^{\ast }\neq 0,$ and/or $h_{4},$ for $h_{4}^{\ast }\neq 0.$
The gravitational polarizations $\eta _{1}$ and $\eta _{2},$ when $\eta
_{1}=\eta _{2}r^{2}=e^{\psi (\xi ,\vartheta )},$ \ are found from (\ref{eq1}%
) with zero source, written in the form $\psi ^{\bullet \bullet }+\psi
^{\prime \prime }=0.$

Introducing the defined values of the coefficients in the ansatz (\ref{5sol1}%
), we find a class of exact off--diagonal vacuum solutions of the Einstein
equations defining stationary nonholonomic deformations of the
Sch\-warz\-schild metric, {\small
\begin{eqnarray}
~^{\varepsilon }\mathbf{g} &=&-e^{\psi }\left( d\xi \otimes d\xi +\
d\vartheta \otimes d\vartheta \right) -4\left[ (\sqrt{|\eta _{4}|}) ^{\ast }%
\right] ^{2}\varpi ^{2}\ \delta \varphi \otimes \ \delta \varphi +\eta
_{4}\varpi ^{2}\ \delta t\otimes \delta t,  \notag \\
\delta \varphi &=&d\varphi +w_{1}d\xi +w_{2} d\vartheta ,\ \delta t=dt+ \
^{1}n_{1}d\xi +\ ^{1}n_{2}d\vartheta.  \label{5sol1a}
\end{eqnarray}%
} The N--connection coefficients $w_{i}(\xi ,\vartheta ,\varphi)$ and $\
^{1}n_{i}(\xi,\vartheta)$ must satisfy the conditions (\ref{vaclcsoc}) in
order to get vacuum metrics in GR.

It should be emphasized here that, in general, the bulk of solutions from
the set of target metrics do not define black holes and do not describe
obvious physical situations. They preserve the singular character of the
coefficient $\varpi ^{2}$ vanishing on the horizon of a Schwarzschild black
hole if we take only smooth integration functions for some small deformation
parameters $\varepsilon .$

\subsection{Linear parametric polarizations induced by $f$--modifi\-cations}

We may select some locally anisotropic configurations with possible physical
interpretation of gravitational vacuum configurations with spherical and/or
rotoid (ellipsoid) symmetry if it is considered a generating function
\begin{equation}
b^{2}=q(\xi ,\vartheta ,\varphi )+\varepsilon \varrho (\xi ,\vartheta
,\varphi ).  \label{gf1}
\end{equation}%
For simplicity, we restrict our analysis only to linear decompositions on
a small parameter $\varepsilon ,$ with $0<\varepsilon <<1.$

Using (\ref{gf1}), we compute $\left( b^{\ast }\right) ^{2}=[(\sqrt{|q|}%
)^{\ast }]^{2}\ [1+\varepsilon \frac{1}{(\sqrt{|q|})^{\ast }}({\varrho }/%
\sqrt{|q|})^{\ast }]$ and the vertical coefficients of d--metric (\ref%
{5sol1a}), i.e., $h_{3}$ and $h_{4}$ (and corresponding polarizations $\eta
_{3}$ and $\eta _{4}),$ see formulas (\ref{aux41})\footnote{%
Nonholonomic deformations of the Schwarzschild solution (not depending on $%
\varepsilon )$ can be generated if we consider $\varepsilon =0$ and $b^{2}=q$
and $\left( b^{\ast }\right) ^{2}=\left[ (\sqrt{|q|})^{\ast }\right] ^{2}.$}.
We model rotoid configurations \ if we chose
\begin{equation}
q=1-\frac{2\mu (\xi ,\vartheta ,\varphi )}{r}\mbox{ and }\varrho =\frac{%
q_{0}(r)}{4\mu ^{2}}\sin (\omega _{0}\varphi +\varphi _{0}),  \label{aux42}
\end{equation}%
for $\mu (\xi ,\vartheta ,\varphi )=\mu _{0}+\varepsilon \mu _{1}(\xi
,\vartheta ,\varphi )$ (supposing that the mass is locally anisotropically
polarized) with certain constants $\mu ,\omega _{0}$ and $\varphi _{0}$ and
arbitrary functions/ polarizations $\mu _{1}(\xi ,\vartheta ,\varphi )$ and $%
q_{0}(r)$ to be determined from some boundary conditions, with $\varepsilon $
being the eccentricity.\footnote{%
We may treat $\varepsilon $ as an eccentricity imposing the condition that
the coefficient $h_{4}=b^{2}=\eta _{4}(\xi ,\vartheta ,\varphi )\varpi
^{2}(\xi )$ becomes zero for data (\ref{aux42}), if $r_{+}\simeq 2\mu
_{0}/\left( 1+\varepsilon \frac{q_{0}(r)}{4\mu ^{2}}\sin (\omega _{0}\varphi
+\varphi _{0})\right) .$}. This condition defines a small deformation of
the Schwarzschild spherical horizon into an ellipsoidal one (rotoid
configuration with eccentricity $\varepsilon ).$

The resulting off--diagonal solution with rotoid type symmetry is
{\small
\begin{eqnarray}
~^{rot}\mathbf{g} &=&-e^{\psi }\left( d\xi \otimes d\xi +\ d\vartheta
\otimes d\vartheta \right) +\left( q+\varepsilon \varrho \right) \ \delta
t\otimes \delta t  -4[(\sqrt{|q|})^{\ast }] ^{2}\ [1+\varepsilon \frac{1}{(\sqrt{%
|q|})^{\ast }} ({\varrho }/\sqrt{|q|}) ^{\ast }] \ \delta \varphi \otimes \
\delta \varphi ,  \notag \\
\delta \varphi &=&d\varphi +w_{1}d\xi +w_{2}d\vartheta ,\ \delta t=dt+\
^{1}n_{1}d\xi +\ ^{1}n_{2}d\vartheta . \label{rotoidm}
\end{eqnarray}%
}
The functions $q(\xi ,\vartheta ,\varphi )$ and $\varrho (\xi ,\vartheta
,\varphi )$ are given by formulas (\ref{aux42}) and the N--connection
coefficients $w_{i}(\xi ,\vartheta ,\varphi )$ and $\ n_{i}=$ $\
^{1}n_{i}(\xi ,\vartheta )$ are subjected to conditions of type (\ref%
{vaclcsoc}),
\begin{eqnarray}
w_{1}w_{2}\left( \ln |\frac{w_{1}}{w_{2}}|\right) ^{\ast } &=&w_{2}^{\bullet
}-w_{1}^{\prime },\quad w_{i}^{\ast }\neq 0;  \label{constr3} \\
\mbox{ or \ }w_{2}^{\bullet }-w_{1}^{\prime } &=&0,\quad w_{i}^{\ast }=0;\
^{1}n_{1}^{\prime }(\xi ,\vartheta )-\ ^{1}n_{2}^{\bullet }(\xi ,\vartheta
)=0  \notag
\end{eqnarray}%
and $\psi (\xi ,\vartheta )$ being any function for which $\psi ^{\bullet
\bullet }+\psi ^{\prime \prime }=0.$

\section{Ellipsoidal $f$--Configurations and Solitons}

\label{s7} We can consider nonholonomic deformations in modified gravity for
arbitrary signs of the cosmological constant $\Lambda \neq 0$ containing
contributions from nonholonomic $f$--modifications. Such classes of
solutions can be constructed in general form for a system (\ref{eq1})--(\ref%
{eq4}) and (\ref{lccondm}) with coefficients of metric of type (\ref{sol1}).
Such metrics consist a particular case of non--vacuum ansatz defined by
Corollary \ref{corym1} with $\underline{h}_{4}=1$ and $\omega =1.$

\subsection{Nonholonomic rotoid deformations}

Let us consider a diagonal metric of type
\begin{equation}
~_{\lambda }^{\varepsilon }\mathbf{g}=d\xi \otimes d\xi +r^{2}(\xi )\
d\theta \otimes d\theta +r^{2}(\xi )\sin ^{2}\theta \ d\varphi \otimes
d\varphi +\ _{\lambda }\varpi ^{2}(\xi )\ dt\otimes \ dt,  \label{sds1}
\end{equation}%
where nontrivial metric coefficients are parametriz\-ed in the form $\check{g%
}_{1}=1,\ \check{g}_{2}=r^{2}(\xi ),\ \check{h}_{3}=r^{2}(\xi )\sin
^{2}\vartheta ,\ \check{h}_{4}=\ _{\lambda }\varpi ^{2}(\xi )$, for local
coordinates $x^{1}=\xi ,x^{2}=\vartheta ,y^{3}=\varphi ,y^{4}=t,$ with $\xi
=\int $ $dr/\left| q(r)\right| ^{\frac{1}{2}},$ and $\ _{\lambda }\varpi
^{2}(r)=-\sigma ^{2}(r)q(r),$ for $q(r)=1-2m(r)/r-\Lambda r^{2}/3.$

The ansatz for such classes of solutions is chosen to be of the form {\small
\begin{eqnarray*}
\ ^{\lambda }\mathbf{\mathring{g}} &=&e^{\underline{\phi }(\xi ,\theta )}\
(d\xi \otimes d\xi +\ d\theta \otimes d\theta )+h_{3}(\xi ,\theta ,\varphi
)\ {\delta }\varphi \otimes {\delta }\varphi +h_{4}(\xi ,\theta ,\varphi )\ {%
\delta t}\otimes ~{\delta t}, \\
~\delta \varphi &=&d\varphi +w_{1}\left( \xi ,\theta ,\varphi \right) d\xi
+w_{2}\left( \xi ,\theta ,\varphi \right) d\theta ,\ \delta t =dt+n_{1}\left( \xi ,\theta ,\varphi \right) d\xi +n_{2}\left(
\xi ,\theta ,\varphi \right) d\theta ,
\end{eqnarray*}%
} for $h_{3}=-h_{0}^{2}(b^{\ast })^{2}=\eta _{3}(\xi ,\theta ,\varphi
)r^{2}(\xi )\sin ^{2}\vartheta ,\ h_{4}=b^{2}=\eta _{4}(\xi ,\theta ,\varphi
)\ _{\lambda }\varpi ^{2}(\xi ).$ The coefficients \ of this metric
determine exact solutions if
\begin{eqnarray}
&&\underline{\phi }^{\bullet \bullet }(\xi ,\theta )+\underline{\phi }%
^{^{\prime \prime }}(\xi ,\theta )=2\Lambda ;  \label{anhsol2} \\
h_{3} &=&\pm \frac{\left( \phi ^{\ast }\right) ^{2}}{4\Lambda }e^{-2\
^{0}\phi (\xi ,\theta )},\ h_{4}=\mp \frac{1}{4\ \Lambda }e^{2(\phi -\
^{0}\phi (\xi ,\theta ))};\
w_{i} =-\partial _{i}\phi /\phi ^{\ast };  \notag \\
n_{i} &=&\ ^{1}n_{i}(\xi ,\theta )+\ ^{2}n_{i}(\xi ,\theta )\int \left( \phi
^{\ast }\right) ^{2}e^{-2(\phi -\ ^{0}\phi (\xi ,\vartheta ))}d\varphi ,\
\notag \\
&=&\ \ \left\{
\begin{array}{rcl}
\ \ ^{1}n_{i}(\xi ,\theta )+\ ^{2}n_{i}(\xi ,\theta )\int e^{-4\phi }\frac{%
\left( h_{4}^{\ast }\right) ^{2}}{h_{4}}d\varphi ,\ \  & \mbox{ if \ } &
n_{i}^{\ast }\neq 0; \\
\ ^{1}n_{i}(\xi ,\theta ),\quad \qquad \qquad \qquad \qquad \qquad &
\mbox{
if \ } & n_{i}^{\ast }=0;%
\end{array}%
\right.  \notag
\end{eqnarray}%
for any nonzero $h_{a}$ and $h_{a}^{\ast }$ and (integrating) functions $%
^{1}n_{i}(\xi ,\theta ),\ ^{2}n_{i}(\xi ,\theta ),$ generating function $%
\phi (\xi ,\theta ,\varphi )$, and $\ ^{0}\phi (\xi ,\theta )$ to be
determined from certain boundary conditions for a fixed system of
coordinates.

For nonholonomic effective ellipsoid de Sitter configurations, we  parameterize
\begin{eqnarray}
~_{\lambda }^{rot}\mathbf{g} &=&-e^{\underline{\phi }(\xi ,\theta )}\left(
d\xi \otimes d\xi +\ d\theta \otimes d\theta \right) +\left( \underline{q}%
+\varepsilon \underline{\varrho }\right) \ \delta t\otimes \delta t  \notag \\
&& -h_{0}^{2}\left[ (\sqrt{|\underline{q}|})^{\ast }\right]
^{2}[1+\varepsilon \frac{1}{(\sqrt{|\underline{q}|})^{\ast }}(\underline{%
\varrho }/\sqrt{|\underline{q}|})^{\ast }]\ \delta \varphi \otimes \ \delta
\varphi ,  \notag \\
\delta \varphi &=&d\varphi +w_{1}d\xi +w_{2}d\vartheta ,\ \delta
t=dt+n_{1}d\xi +n_{2}d\vartheta ,  \label{soladel}
\end{eqnarray}%
where $\underline{q}=1-\frac{2\ ^{1}\underline{\mu }(r,\theta ,\varphi )}{r}%
,\ \ \underline{\varrho }=\frac{\underline{q}_{0}(r)}{4\underline{\mu }%
_{0}^{2}}\sin (\omega _{0}\varphi +\varphi _{0})$, are chosen to generate an
anisotro\-pic rotoid configuration for the smaller "horizon" (when $\ h_{4}=0),\ r_{+}\simeq 2\ ^{1}\underline{%
\mu }/\left( 1+\varepsilon \frac{\underline{q}_{0}(r)}{4\underline{\mu }%
_{0}^{2}}\sin (\omega _{0}\varphi +\varphi _{0})\right) ,$ for a
corresponding $\underline{q}_{0}(r).$

We have to impose the condition that the coefficients of the above d--metric
induce a zero torsion in order to generate solutions of the Einstein
equations for the Levi--Civita connection. Using formula (\ref{constr3}),
for $\phi ^{\ast }\neq 0,$ we obtain that $\phi (r,\varphi ,\theta )=\ln
|h_{4}^{\ast }/\sqrt{|h_{3}h_{4}|}|$ must be any function defined in
non--explicit form from equation $2e^{2\phi }\phi =\Lambda.$ The set of
constraints for the N--connection coefficients is solved, if the integration
functions in (\ref{anhsol2}) are chosen in a form when $w_{1}w_{2}\left( \ln
|\frac{w_{1}}{w_{2}}|\right) ^{\ast }=w_{2}^{\bullet }-w_{1}^{\prime }$ for $%
w_{i}^{\ast }\neq 0;$ $w_{2}^{\bullet }-w_{1}^{\prime }=0$ for $\
w_{i}^{\ast }=0;$ and take $\ n_{i}=\ ^{1}n_{i}(x^{k})$ for $\
^{1}n_{1}^{\prime }(x^{k})-\ ^{1}n_{2}^{\bullet }(x^{k})=0.$

In a particular case, in the limit $\varepsilon \rightarrow 0,$ we get a
subclass of solutions of type (\ref{soladel}) with spherical symmetry but
with generic off--diagonal coefficients induced by the N--connection
coefficients. This class of spacetimes depend on cosmological constants
polarized nonholonomically by $f$--modifications. We can extract from such
configurations the Schwarzschild solution, if we select a set of functions
with the properties $\phi \rightarrow const,w_{i}\rightarrow
0,n_{i}\rightarrow 0$ and $h_{4}\rightarrow \varpi ^{2}.$

\subsection{Effective vacuum solitonic configurations}

It is possible to construct off--diagonal vacuum spacetimes generated by
3--d solitons as examples of generic off--diagonal solutions with nontrivial
vertical conformal factor $\omega .$ We consider that there are satisfied
the conditions of Corollary \ref{corym2} with $\underline{h}_{4}=1$ for
effective vacuum solutions (such configurations may encode $f$%
--modifications) and the Cauchy problem is stated as in Section \ref{s3}.

\subsubsection{Solutions with solitonic factor $\omega (x^{1},y^{3},t)$}

We take $\omega =\eta (x^{1},y^{3},t),$ when $y^{4}=t$ is a timelike
coordinate, as a solution of KdP equation \cite{kadom}
\begin{equation}
\pm \eta ^{\ast \ast }+(\partial _{t}\eta +\eta \ \eta ^{\bullet }+\epsilon
\eta ^{\bullet \bullet \bullet })^{\bullet }=0,  \label{kdp1}
\end{equation}%
with dispersion $\epsilon $ and possible dependencies on a set of parameters
$\theta .$ It is supposed that, in the dispersionless limit $\epsilon
\rightarrow 0$, the solutions are independent on $y^{3}$ and determined by
Burgers' equation $\partial _{t}\eta +\eta \ \eta ^{\bullet }=0.$ For such
3--d solitonic configurations, the conditions (\ref{lccondm}) are written in
the form%
\begin{equation*}
\mathbf{e}_{1}\eta =\eta ^{\bullet }+w_{1}(x^{i},y^{3})\eta ^{\ast
}+n_{1}(x^{i})\partial _{t}\eta =0.
\end{equation*}%
For $\eta ^{\prime }=0,$ we can impose the condition $w_{2}=0$ and $n_{2}=0.$

Such vacuum solitonic metrics can be parametrized in the form
\begin{eqnarray*}
\ _{}\mathbf{g} &=&e^{\psi (x^{k})}(dx^{1}\otimes dx^{1}+dx^{2}\otimes
dx^{2})+\left[ \eta (x^{1},y^{3},t)\right] ^{2}h_{a}(x^{1},y^{3})\ \mathbf{e}%
^{a}\otimes \mathbf{e}^{a}, \\
\mathbf{e}^{3} &=&dy^{3}+w_{1}(x^{k},y^{3})dx^{1},\ \mathbf{e}%
^{4}=dy^{4}+n_{1}(x^{k})dx^{1}.
\end{eqnarray*}%
This class of metrics does not have (in general) Killing \ symmetries but
may possess symmetries determined by solitonic solutions of (\ref{kdp1}).
Alternatively, we can consider\ that $\eta $ is a solution of any three
dimensional solitonic and/ or other nonlinear wave equations; in a similar
manner, we can generate solutions for $\omega =\eta (x^{2},y^{3},t).$

\subsubsection{Solitonic metrics with factor $\protect\omega (x^{i},t)$}

There are effective vacuum metrics when the solitonic dynamics does not
depend on anisotropic coordinate $y^{3}.$ In this case $\omega =\widehat{%
\eta }(x^{k},t)$ is a solution of KdP equation
\begin{equation}
\pm \widehat{\eta }^{\bullet \bullet }+(\partial _{t}\widehat{\eta }+%
\widehat{\eta }\ \widehat{\eta }^{\prime }+\epsilon \widehat{\eta }^{\prime
\prime \prime })^{\prime }=0.  \label{kdp3a}
\end{equation}%
In the dispersionless limit $\epsilon \rightarrow 0$ the solutions are
independent on $x^{1}$ and determined by Burgers' equation $\partial _{t}%
\widehat{\eta }+\widehat{\eta }\ \widehat{\eta }^{\prime }=0.$

This class of vacuum solitonic modified gravity configurations is given by
\begin{eqnarray*}
\ _{2}\mathbf{g} &=&e^{\psi (x^{k})}(dx^{1}\otimes dx^{1}+dx^{2}\otimes
dx^{2})+\left[ \widehat{\eta }(x^{k},t)\right] ^{2}h_{a}(x^{k},y^{3})\
\mathbf{e}^{a}\otimes \mathbf{e}^{a}, \\
\mathbf{e}^{3} &=&dy^{3}+w_{1}(x^{k},y^{3})dx^{1},\ \mathbf{e}%
^{4}=dy^{4}+n_{1}(x^{k})dx^{1};
\end{eqnarray*}%
the conditions (\ref{lccondm}) are $\mathbf{e}_{1}\widehat{\eta }=\widehat{%
\eta }^{\bullet }+n_{1}(x^{i})\partial _{t}\widehat{\eta }=0,\ \mathbf{e}_{2}%
\widehat{\eta }=\widehat{\eta }^{\prime }+n_{2}(x^{i})\partial _{t}\widehat{%
\eta }=0$.

It is possible to derive an infinite number of vacuum gravitational 2-d and
3-d configurations characterized by corresponding solitonic hierarchies and
bi--Hamilton structures, for instance, related to different KdP equations (%
\ref{kdp3a}) with possible mixtures with solutions for 2-d and 3-d
sine--Gordon equations etc, see details in Ref. \cite{vacarsolitonhier}.

\vskip5pt

\textbf{Acknowledgments:\ } The work is partially supported by the Program
IDEI, PN-II-ID-PCE-2011-3-0256. Author is grateful to T. Grammenos, A. P.
Kouretisis, N. E. Mavromatos, P. C. Stavrinos and E. C. Vagenas for
important discussions, kind support and hospitality.

\appendix

\setcounter{equation}{0} \renewcommand{\theequation}
{A.\arabic{equation}} \setcounter{subsection}{0}
\renewcommand{\thesubsection}
{A.\arabic{subsection}}

\section{Proof of Theorem \protect\ref{th2a}}

\label{sb}For $\omega =1$ and $\underline{h}_{a}=const,$ such proofs can be
obtained by straightforward computations \cite{ijgmmp}. The approach was
extended for $\omega \neq 1$ and higher dimensions in \cite{vexsol1,vexsol2}%
. In this section, we sketch a proof for ansatz (\ref{ans1}) with nontrivial
$\underline{h}_{4}$ depending on variable $y^{4}$ when $\omega =1$ in data (%
\ref{paramdcoef}). At the next step, the formulas will be completed for
nontrivial values $\omega \neq 1$.

If $\widehat{R}_{~1}^{1}=\widehat{R}_{~2}^{2}$ and $\widehat{R}_{~3}^{3}=%
\widehat{R}_{~4}^{4},$ the effective Einstein equations (\ref{cdeinst}) for $%
\widehat{\mathbf{D}}$ and data (\ref{data1a}) (see below) can be written for
any source (\ref{source}) in the form {\small
\begin{equation*}
\widehat{E}_{~1}^{1}=\widehat{E}_{~2}^{2}=-\widehat{R}_{~3}^{3}=\mathbf{%
\Upsilon }(x^{k},y^{3})+\underline{\mathbf{\Upsilon }}(x^{k},y^{3},y^{4}),\
\widehat{E}_{~3}^{3}=\widehat{E}_{~4}^{4}=-\widehat{R}_{~1}^{1}=\ ^{v}%
\mathbf{\Upsilon }(x^{k}).
\end{equation*}%
} The geometric data for the conditions of Theorem \ref{th2a} are $%
g_{i}=g_{i}(x^{k})$ and {\small
\begin{equation}
g_{3}=h_{3}(x^{k},y^{3}),g_{4}=h_{4}(x^{k},y^{3})\underline{h}%
_{4}(x^{k},y^{4}),N_{i}^{3}=w_{i}(x^{k},y^{3}),N_{i}^{4}=n_{i}(x^{k},y^{3}),
\label{data1a}
\end{equation}%
} for $\underline{h}_{3}=1$ and local coordinates $u^{\alpha
}=(x^{i},y^{a})=(x^{1},x^{2},y^{3},y^{4}).$ \ For such values, we shall
compute respectively the coefficients of $\Omega _{\ \alpha \beta }^{\ a}$, canonical d--connection $\widehat{\mathbf{\Gamma }}_{\
\alpha \beta }^{\gamma }$ (\ref{cdc}), d--torsion $\widehat{\mathbf{T}}_{\
\alpha \beta }^{\gamma }$\ (\ref{dtors}), necessary coefficients of
d--curvature $\widehat{\mathbf{R}}_{\ \alpha \beta \gamma }^{\tau }$ (\ref%
{dcurv}) with respective contractions for $\widehat{\mathbf{R}}_{\alpha
\beta }:=\widehat{\mathbf{R}}_{\ \alpha \beta \gamma }^{\gamma }$ (\ref%
{driccic}) and resulting $\ ^{s}\widehat{R}$ (\ref{sdcurv}) and $\widehat{%
\mathbf{E}}_{\alpha \beta }$(\ref{enstdt}). Finally, we shall state the
conditions (\ref{lcconstr}) when general coefficients (\ref{data1a}) are
considered for d--metrics.

\subsection{Coefficients of the canonical d--connection}

There are horizontal nontrivial coefficients of $\widehat{\mathbf{\Gamma }}%
_{\ \alpha \beta }^{\gamma }$ (\ref{cdc}),{\small
\begin{eqnarray*}
\widehat{L}_{jk}^{i} &=&\frac{1}{2}g^{i1}\left( \mathbf{e}_{k}g_{j1}+\mathbf{%
e}_{j}g_{k1}-\mathbf{e}_{1}g_{jk}\right) +\frac{1}{2}g^{i2}\left( \mathbf{e}%
_{k}g_{j2}+\mathbf{e}_{j}g_{k2}-\mathbf{e}_{2}g_{jk}\right),  \\
\mbox{ i.e., }\widehat{L}_{jk}^{1} &=&\frac{1}{2g_{1}}\left( \mathbf{\partial
}_{k}g_{j1}+\mathbf{\partial }_{j}g_{k1}-g_{jk}^{\bullet }\right) ,\widehat{L%
}_{jk}^{2}=\frac{1}{2g_{2}}\left( \mathbf{\partial }_{k}g_{j2}+\mathbf{%
\partial }_{j}g_{k2}-g_{jk}^{\prime }\right) .
\end{eqnarray*}%
} The h--v--components $\widehat{L}_{bk}^{a}$ are computed following
formulas {\small
\begin{eqnarray*}
\widehat{L}_{bk}^{3} &=&e_{b}(N_{k}^{3})+\frac{1}{2g_{3}}[\mathbf{e}%
_{k}g_{b3}-g_{3}\ e_{b}N_{k}^{3}-g_{3b}(N_{k}^{3})^{\ast
}-g_{4b}(N_{k}^{4})^{\circ }]=e_{b}(N_{k}^{3}) \\
&&+\frac{1}{2g_{3}}[\mathbf{\partial }_{k}g_{b3}-N_{k}^{3}g_{b3}^{\ast
}-N_{k}^{4}g_{b3}^{\circ }-g_{3}\ e_{b}N_{k}^{3}-g_{3b}(N_{k}^{3})^{\ast
}-g_{4b}(N_{k}^{4})^{\ast }] \\
&=&e_{b}(w_{k}+\underline{w}_{k})+\frac{1}{2g_{3}}[\mathbf{\partial }%
_{k}g_{b3}-(w_{k}+\underline{w}_{k})g_{b3}^{\ast }-(n_{k}+\underline{n}%
_{k})g_{b3}^{\circ } -g_{3}\ e_{b}(w_{k}+\underline{w}_{k})-g_{3b}w_{k}{}^{\ast
}-g_{4b}n_{k}{}^{\ast }],
\end{eqnarray*}%
\begin{eqnarray*}
\widehat{L}_{bk}^{4} &=&e_{b}(N_{k}^{4})+\frac{1}{2g_{4}}[\mathbf{e}%
_{k}g_{b4}-g_{4}\ e_{b}N_{k}^{4}-g_{3b}(N_{k}^{3})^{\circ
}-g_{4b}(N_{k}^{4})^{\circ }]=e_{b}(N_{k}^{4}) \\
&&+\frac{1}{2g_{4}}[\mathbf{\partial }_{k}g_{b4}-N_{k}^{3}g_{b4}^{\ast
}-N_{k}^{4}g_{b4}^{\circ }-g_{4}\ e_{b}N_{k}^{4}-g_{3b}(N_{k}^{3})^{\circ
}-g_{4b}(N_{k}^{4})^{\circ }] \\
&=&e_{b}(n_{k}+\underline{n}_{k})+\frac{1}{2g_{4}}[\mathbf{\partial }%
_{k}g_{b4}-(w_{k}+\underline{w}_{k})g_{b4}^{\ast }-(n_{k}+\underline{n}%
_{k})g_{b4}^{\circ } -g_{4}\ e_{b}(n_{k}+\underline{n}_{k})-g_{3b}\underline{w}_{k}^{\circ
}-g_{4b}\underline{n}_{k}^{\circ }].
\end{eqnarray*}%
} In explicit form, we obtain these nontrivial values{\small
\begin{eqnarray*}
\widehat{L}_{3k}^{3} &=&w_{k}^{\ast }+\frac{1}{2g_{3}}[\mathbf{\partial }%
_{k}g_{3}-w_{k}g_{3}^{\ast }-n_{k}g_{3}^{\circ }-g_{3}\ w_{k}^{\ast
}-g_{3}w_{k}{}^{\ast }] = \frac{1}{2g_{3}}[\mathbf{\partial }_{k}g_{3}-w_{k}g_{3}^{\ast }]=\frac{%
\mathbf{\partial }_{k}h_{3}}{2h_{3}}-w_{k}\frac{h_{3}^{\ast }}{2h_{3}}, \\
\widehat{L}_{4k}^{3} &=&\frac{1}{2g_{3}}[-g_{4}n_{k}{}^{\ast }]=-\frac{h_{4}%
\underline{h}_{4}}{2h_{3}}n_{k}{}^{\ast },\ \widehat{L}_{3k}^{4}=n_{k}^{\ast
}+\frac{1}{2g_{4}}[-g_{4}n_{k}^{\ast }]=\frac{1}{2}n_{k}^{\ast },\\
\widehat{L}_{4k}^{4}&=&\frac{1}{2g_{4}}[\mathbf{\partial }%
_{k}g_{4}-w_{k}g_{4}^{\ast }-n_{k}g_{4}^{\circ }]=\frac{\mathbf{\partial }%
_{k}(h_{4}\underline{h}_{4})}{2h_{4}\underline{h}_{4}}-w_{k}\frac{%
h_{4}^{\ast }}{2h_{4}}-n_{k}\frac{\underline{h}_{4}^{\circ }}{2\underline{h}%
_{4}}.
\end{eqnarray*}
} For the set of $h$--$v$ $C$--coefficients, we get $\widehat{C}_{jc}^{i}=%
\frac{1}{2}g^{ik}\frac{\partial g_{jk}}{\partial y^{c}}=0$. The $v$%
--components of $C$--coefficients are computed using the following formulas
\begin{equation*}
\widehat{C}_{bc}^{3}=\frac{1}{2g_{3}}\left(
e_{c}g_{b3}+e_{c}g_{c3}-g_{bc}^{\ast }\right) ,\ \widehat{C}_{bc}^{4}=\frac{1%
}{2g_{4}}\left( e_{c}g_{b4}+e_{b}g_{c4}-g_{bc}^{\circ }\right) ,
\end{equation*}%
i. e., $\widehat{C}_{33}^{3}=\frac{g_{3}^{\ast }}{2g_{3}}=\frac{h_{3}^{\ast }%
}{2h_{3}},~\widehat{C}_{34}^{3}=\frac{g_{3}^{\circ }}{2g_{3}}=0,\ \widehat{C}%
_{44}^{3}=-\frac{g_{4}^{\ast }}{2g_{3}}=-\frac{h_{4}^{\ast }\underline{h}_{4}%
}{h_{3}},\widehat{C}_{33}^{4}=-\frac{g_{3}^{\circ }}{2g_{4}}=0,~\widehat{C}%
_{34}^{4}=\frac{g_{4}^{\ast }}{2g_{4}}=\frac{h_{4}^{\ast }}{2h_{4}},\widehat{%
C}_{44}^{4}=\frac{g_{4}^{\circ }}{2g_{4}}=\frac{\underline{h}_{4}^{\circ }}{2%
\underline{h}_{4}}$.

Putting together the above formulas, we find all
nontrivial coefficients, {\small
\begin{eqnarray}
\widehat{L}_{11}^{1} &=&\frac{g_{1}^{\bullet }}{2g_{1}},\ \widehat{L}%
_{12}^{1}=\frac{g_{1}^{\prime }}{2g_{1}},\widehat{L}_{22}^{1}=-\frac{%
g_{2}^{\bullet }}{2g_{1}},\ \widehat{L}_{11}^{2}=\frac{-g_{1}^{\prime }}{%
2g_{2}},\ \widehat{L}_{12}^{2}=\frac{g_{2}^{\bullet }}{2g_{2}},\ \widehat{L}%
_{22}^{2}=\frac{g_{2}^{\prime }}{2g_{2}},  \label{nontrdc} \\
\widehat{L}_{4k}^{4} &=&\frac{\mathbf{\partial }_{k}(h_{4}\underline{h}_{4})%
}{2h_{4}\underline{h}_{4}}-\frac{w_{k}h_{4}^{\ast }}{2h_{4}}-(n_{k}+%
\underline{n}_{k})\frac{\underline{h}_{4}^{\circ }}{2\underline{h}_{4}},%
\widehat{L}_{3k}^{3}=\frac{\mathbf{\partial }_{k}h_{3}}{2h_{3}}-\frac{%
w_{k}h_{3}^{\ast }}{2h_{3}},\widehat{L}_{4k}^{3}=\frac{h_{4}\underline{h}_{4}%
}{-2h_{3}}n_{k}^{\ast },  \notag \\
\widehat{L}_{3k}^{4} &=&\frac{1}{2}n_{k}^{\ast },\widehat{C}_{33}^{3}=\frac{%
h_{3}^{\ast }}{2h_{3}},\widehat{C}_{44}^{3}=-\frac{h_{4}^{\ast }\underline{h}%
_{4}}{h_{3}},\ \widehat{C}_{33}^{4}=-\frac{h_{3}\underline{h}_{3}^{\circ }}{%
h_{4}\underline{h}_{4}},~\widehat{C}_{34}^{4}=\frac{h_{4}^{\ast }}{2h_{4}},%
\widehat{C}_{44}^{4}=\frac{\underline{h}_{4}^{\circ }}{2\underline{h}_{4}}.
\notag
\end{eqnarray}%
} We shall need also the values
\begin{equation}
\ \widehat{C}_{3}=\widehat{C}_{33}^{3}+\widehat{C}_{34}^{4}=\frac{%
h_{3}^{\ast }}{2h_{3}}+\frac{h_{4}^{\ast }}{2h_{4}},\widehat{C}_{4}=\widehat{%
C}_{43}^{3}+\widehat{C}_{44}^{4}=\frac{\underline{h}_{4}^{\circ }}{2%
\underline{h}_{4}}.  \label{aux3}
\end{equation}

\subsection{Coefficients for torsion of $\widehat{\mathbf{D}}$}

Using data (\ref{data1a}) for $\underline{w}_{i}=\underline{n}_{i}=0,$ the
coefficients $\Omega _{ij}^{a}=\mathbf{e}_{j}\left( N_{i}^{a}\right) -%
\mathbf{e}_{i}(N_{j}^{a})$, are computed
\begin{eqnarray*}
\Omega _{ij}^{a} &=&\mathbf{\partial }_{j}\left( N_{i}^{a}\right) -\partial
_{i}(N_{j}^{a})-N_{i}^{b}\partial _{b}N_{j}^{a}+N_{j}^{b}\partial
_{b}N_{i}^{a} \\
&=&\mathbf{\partial }_{j}\left( N_{i}^{a}\right) -\partial
_{i}(N_{j}^{a})-N_{i}^{3}(N_{j}^{a})^{\ast }-N_{i}^{4}(N_{j}^{a})^{\circ
}+N_{j}^{3}(N_{i}^{a})^{\ast }+N_{j}^{4}(N_{i}^{a})^{\circ } \\
&=&\mathbf{\partial }_{j}\left( N_{i}^{a}\right) -\partial
_{i}(N_{j}^{a})-w_{i}(N_{j}^{a})^{\ast }+w_{j}(N_{i}^{a})^{\ast },
\end{eqnarray*}%
and the nontrivial values for coefficients are  {\small
\begin{eqnarray}
\Omega _{12}^{3} &=&-\Omega _{21}^{3}=\mathbf{\partial }_{2}w_{1}-\partial
_{1}w_{2}-w_{1}w_{2}^{\ast }+w_{2}w_{1}^{\ast }=w_{1}^{\prime
}-w_{2}^{\bullet }-w_{1}w_{2}{}^{\ast }+w_{2}w_{1}^{\ast }{};  \notag \\
\Omega _{12}^{4} &=&-\Omega _{21}^{4}=\mathbf{\partial }_{2}n_{1}-\partial
_{1}n_{2}-w_{1}n_{2}^{\ast }+w_{2}n_{1}^{\ast }=n_{1}^{\prime
}-n_{2}^{\bullet }-w_{1}n_{2}^{\ast }{}+w_{2}n_{1}^{\ast }{}.  \label{omeg}
\end{eqnarray}%
} The nontrivial coefficients of d--torsion (\ref{dtors}) are $\widehat{T}%
_{\ ji}^{a}=-\Omega _{\ ji}^{a}$ (\ref{omeg}) and $\widehat{T}_{aj}^{c}=%
\widehat{L}_{aj}^{c}-e_{a}(N_{j}^{c}).$ We find for other types of
coefficients that
\begin{equation}
\widehat{T}_{\ jk}^{i}=\widehat{L}_{jk}^{i}-\widehat{L}_{kj}^{i}=0,~\widehat{%
T}_{\ ja}^{i}=\widehat{C}_{jb}^{i}=0,~\widehat{T}_{\ bc}^{a}=\ \widehat{C}%
_{bc}^{a}-\ \widehat{C}_{cb}^{a}=0.  \notag
\end{equation}%
We have such nontrivial N--adapted coefficients of d--torsion {\small
\begin{eqnarray}
\widehat{T}_{3k}^{3} &=&\widehat{L}_{3k}^{3}-e_{3}(N_{k}^{3})=\frac{\mathbf{%
\partial }_{k}h_{3}}{2h_{3}}-w_{k}\frac{h_{3}^{\ast }}{2h_{3}}-w_{k}^{\ast
}{},  \notag \\
\widehat{T}_{4k}^{3} &=&\widehat{L}_{4k}^{3}-e_{4}(N_{k}^{3})=-\frac{h_{4}%
\underline{h}_{4}}{2h_{3}}n_{k}^{\ast },\ \widehat{T}_{3k}^{4}=~\widehat{L}%
_{3k}^{4}-e_{3}(N_{k}^{4})=\frac{1}{2}n_{k}^{\ast }-n_{k}^{\ast }=-\frac{1}{2%
}n_{k}^{\ast },  \notag \\
\widehat{T}_{4k}^{4} &=&\widehat{L}_{4k}^{4}-e_{4}(N_{k}^{4})=\frac{\mathbf{%
\partial }_{k}(h_{4}\underline{h}_{4})}{2h_{4}\underline{h}_{4}}-w_{k}\frac{%
h_{4}^{\ast }}{2h_{4}}-n_{k}\frac{\underline{h}_{4}^{\circ }}{2\underline{h}%
_{4}},  \notag \\
-\widehat{T}_{12}^{3} &=&w_{1}^{\prime }-w_{2}^{\bullet }-w_{1}w_{2}^{\ast
}{}+w_{2}w_{1}^{\ast },\ -\widehat{T}_{12}^{4}=n_{1}^{\prime
}-n_{2}^{\bullet }-w_{1}n_{2}^{\ast }{}+w_{2}n_{1}^{\ast }{}.
\label{nontrtors}
\end{eqnarray}%
} If all coefficients (\ref{nontrtors}) are zero, then $\Gamma _{\ \alpha \beta
}^{\gamma }=\widehat{\mathbf{\Gamma }}_{\ \alpha \beta }^{\gamma }.$

\subsection{Calculation of the Ricci tensor}

Let us compute the values $\widehat{R}_{ij}=\widehat{R}_{\ ijk}^{k}$ from (%
\ref{riccid}) using (\ref{dcurv}),
{\small
$$\widehat{R}_{\ hjk}^{i} =\mathbf{e}_{k}\widehat{L}_{.hj}^{i}-\mathbf{e}_{j}%
\widehat{L}_{hk}^{i}+\widehat{L}_{hj}^{m}\widehat{L}_{mk}^{i}-\widehat{L}%
_{hk}^{m}\widehat{L}_{mj}^{i}-\widehat{C}_{ha}^{i}\Omega _{jk}^{a} = \mathbf{\partial }_{k}\widehat{L}_{.hj}^{i}-\partial _{j}\widehat{L}%
_{hk}^{i}+\widehat{L}_{hj}^{m}\widehat{L}_{mk}^{i}-\widehat{L}_{hk}^{m}%
\widehat{L}_{mj}^{i},$$
}
where  $\widehat{C}^{i}_{\ ha}=0$ and
$\mathbf{e}_{k}\widehat{L}_{hj}^{i} = \partial _{k}\widehat{L}%
_{hj}^{i}+N_{k}^{a}\partial _{a}\widehat{L}_{hj}^{i} = \partial _{k}\widehat{%
L}_{hj}^{i}+w_{k}\left( \widehat{L}_{hj}^{i}\right) ^{\ast }+n_{k}\left(
\widehat{L}_{hj}^{i}\right) ^{\circ }=\partial _{k}\widehat{L}_{hj}^{i}$.
 Taking derivatives of (\ref{nontrdc}), we obtain {\small
\begin{eqnarray*}
\partial _{1}\widehat{L}_{\ 11}^{1} &=&(\frac{g_{1}^{\bullet }}{2g_{1}}%
)^{\bullet }=\frac{g_{1}^{\bullet \bullet }}{2g_{1}}-\frac{\left(
g_{1}^{\bullet }\right) ^{2}}{2\left( g_{1}\right) ^{2}},\ \partial _{1}%
\widehat{L}_{\ 12}^{1}=(\frac{g_{1}^{\prime }}{2g_{1}})^{\bullet }=\frac{%
g_{1}^{\prime \bullet }}{2g_{1}}-\frac{g_{1}^{\bullet }g_{1}^{\prime }}{%
2\left( g_{1}\right) ^{2}},\  \\
\partial _{1}\widehat{L}_{\ 22}^{1} &=&(-\frac{g_{2}^{\bullet }}{2g_{1}}%
)^{\bullet }=-\frac{g_{2}^{\bullet \bullet }}{2g_{1}}+\frac{g_{1}^{\bullet
}g_{2}^{\bullet }}{2\left( g_{1}\right) ^{2}},\ \partial _{1}\widehat{L}_{\
11}^{2}=(-\frac{g_{1}^{\prime }}{2g_{2}})^{\bullet }=-\frac{g_{1}^{\prime
\bullet }}{2g_{2}}+\frac{g_{1}^{\bullet }g_{2}^{\prime }}{2\left(
g_{2}\right) ^{2}}, \\
\partial _{1}\widehat{L}_{\ 12}^{2} &=&(\frac{g_{2}^{\bullet }}{2g_{2}}%
)^{\bullet }=\frac{g_{2}^{\bullet \bullet }}{2g_{2}}-\frac{\left(
g_{2}^{\bullet }\right) ^{2}}{2\left( g_{2}\right) ^{2}},\ \partial _{1}%
\widehat{L}_{\ 22}^{2}=(\frac{g_{2}^{\prime }}{2g_{2}})^{\bullet }=\frac{%
g_{2}^{\prime \bullet }}{2g_{2}}-\frac{g_{2}^{\bullet }g_{2}^{\prime }}{%
2\left( g_{2}\right) ^{2}},
\end{eqnarray*}
\begin{eqnarray*}
\partial _{2}\widehat{L}_{\ 11}^{1} &=&(\frac{g_{1}^{\bullet }}{2g_{1}}%
)^{\prime }=\frac{g_{1}^{\bullet \prime }}{2g_{1}}-\frac{g_{1}^{\bullet
}g_{1}^{\prime }}{2\left( g_{1}\right) ^{2}},~\partial _{2}\widehat{L}_{\
12}^{1}=(\frac{g_{1}^{\prime }}{2g_{1}})^{\prime }=\frac{g_{1}^{\prime
\prime }}{2g_{1}}-\frac{\left( g_{1}^{\prime }\right) ^{2}}{2\left(
g_{1}\right) ^{2}}, \\
\partial _{2}\widehat{L}_{\ 22}^{1} &=&(-\frac{g_{2}^{\bullet }}{2g_{1}}%
)^{\prime }=-\frac{g_{2}^{\bullet ^{\prime }}}{2g_{1}}+\frac{g_{2}^{\bullet
}g_{1}^{^{\prime }}}{2\left( g_{1}\right) ^{2}},\ \partial _{2}\widehat{L}%
_{\ 11}^{2}=(-\frac{g_{1}^{\prime }}{2g_{2}})^{\prime }=-\frac{g_{1}^{\prime
\prime }}{2g_{2}}+\frac{g_{1}^{\bullet }g_{1}^{\prime }}{2\left(
g_{2}\right) ^{2}}, \\
\partial _{2}\widehat{L}_{\ 12}^{2} &=&(\frac{g_{2}^{\bullet }}{2g_{2}}%
)^{\prime }=\frac{g_{2}^{\bullet \prime }}{2g_{2}}-\frac{g_{2}^{\bullet
}g_{2}^{\prime }}{2\left( g_{2}\right) ^{2}},\partial _{2}\widehat{L}_{\
22}^{2}=(\frac{g_{2}^{\prime }}{2g_{2}})^{\prime }=\frac{g_{2}^{\prime
\prime }}{2g_{2}}-\frac{\left( g_{2}^{\prime }\right) ^{2}}{2\left(
g_{2}\right) ^{2}}.
\end{eqnarray*}
} For these values, there are only two nontrivial components,
\begin{eqnarray*}
\widehat{R}_{\ 212}^{1} &=&\frac{g_{2}^{\bullet \bullet }}{2g_{1}}-\frac{%
g_{1}^{\bullet }g_{2}^{\bullet }}{4\left( g_{1}\right) ^{2}}-\frac{\left(
g_{2}^{\bullet }\right) ^{2}}{4g_{1}g_{2}}+\frac{g_{1}^{\prime \prime }}{%
2g_{1}}-\frac{g_{1}^{\prime }g_{2}^{\prime }}{4g_{1}g_{2}}-\frac{\left(
g_{1}^{\prime }\right) ^{2}}{4\left( g_{1}\right) ^{2}}, \\
\widehat{R}_{\ 112}^{2} &=&-\frac{g_{2}^{\bullet \bullet }}{2g_{2}}+\frac{%
g_{1}^{\bullet }g_{2}^{\bullet }}{4g_{1}g_{2}}+\frac{\left( g_{2}^{\bullet
}\right) ^{2}}{4(g_{2})^{2}}-\frac{g_{1}^{\prime \prime }}{2g_{2}}+\frac{%
g_{1}^{\prime }g_{2}^{\prime }}{4(g_{2})^{2}}+\frac{\left( g_{1}^{\prime
}\right) ^{2}}{4g_{1}g_{2}}.
\end{eqnarray*}%
Considering $\widehat{R}_{11}=-\widehat{R}_{\ 112}^{2}$ and $\widehat{R}%
_{22}=\widehat{R}_{\ 212}^{1},$ when $g^{i}=1/g_{i},$ we find
\begin{equation*}
\widehat{R}_{1}^{1}=\widehat{R}_{2}^{2}=-\frac{1}{2g_{1}g_{2}} [
g_{2}^{\bullet \bullet }-\frac{g_{1}^{\bullet }g_{2}^{\bullet }}{2g_{1}}-%
\frac{\left( g_{2}^{\bullet }\right) ^{2}}{2g_{2}}+g_{1}^{\prime \prime }-%
\frac{g_{1}^{\prime }g_{2}^{\prime }}{2g_{2}}-\frac{(g_{1}^{\prime}) ^{2}}{%
2g_{1}}],
\end{equation*}%
which can be found in equations (\ref{eq2}).

The next step is to derive the equations (\ref{eq3}). We consider the third
formula in (\ref{dcurv}),
{\small
 $$\widehat{R}_{\ bka}^{c} = \frac{\partial \widehat{L}_{bk}^{c}}{\partial
y^{a}}-(\frac{\partial \widehat{C}_{ba}^{c}}{\partial x^{k}}+\widehat{L}%
_{dk}^{c\,}\widehat{C}_{ba}^{d}-\widehat{L}_{bk}^{d}\widehat{C}_{da}^{c}-%
\widehat{L}_{ak}^{d}\widehat{C}_{bd}^{c}) +\widehat{C}_{bd}^{c}\widehat{T}%
_{ka}^{d} = \frac{\partial \widehat{L}_{bk}^{c}}{\partial y^{a}}-\widehat{C}%
_{~ba|k}^{c}+\widehat{C}_{~bd}^{c}\widehat{T}_{~ka}^{d}.$$
}
 Contracting indices, we get $\widehat{R}_{bk}=\widehat{R}_{\ bka}^{a}=\frac{%
\partial L_{bk}^{a}}{\partial y^{a}}-\widehat{C}_{ba|k}^{a}+\widehat{C}%
_{bd}^{a}\widehat{T}_{ka}^{d}$. For $\widehat{C}_{b}:=\widehat{C}_{ba}^{c}$,
we write
$$\widehat{C}_{b|k} =\mathbf{e}_{k}\widehat{C}_{b}-\widehat{L}_{\ bk}^{d\,}%
\widehat{C}_{d}=\partial _{k}\widehat{C}_{b}-N_{k}^{e}\partial _{e}\widehat{C%
}_{b}-\widehat{L}_{\ bk}^{d\,}\widehat{C}_{d} =\partial _{k}\widehat{C}_{b}-w_{k}\widehat{C}_{b}^{\ast }-n_{k}\widehat{C}%
_{b}^{\circ }-\widehat{L}_{\ bk}^{d\,}\widehat{C}_{d}.$$
 We split conventionally $\widehat{R}_{bk}=\ _{[1]}R_{bk}+\ _{[2]}R_{bk}+\
_{[3]}R_{bk},$ where%
\begin{eqnarray*}
\ _{[1]}R_{bk} &=&\left( \widehat{L}_{bk}^{3}\right) ^{\ast }+\left(
\widehat{L}_{bk}^{4}\right) ^{\circ },\ _{[2]}R_{bk}=-\partial _{k}\widehat{C%
}_{b}+w_{k}\widehat{C}_{b}^{\ast }+n_{k}\widehat{C}_{b}^{\circ }+\widehat{L}%
_{\ bk}^{d\,}\widehat{C}_{d}, \\
\ _{[3]}R_{bk} &=&\widehat{C}_{bd}^{a}\widehat{T}_{ka}^{d}=\widehat{C}%
_{b3}^{3}\widehat{T}_{k3}^{3}+\widehat{C}_{b4}^{3}\widehat{T}_{k3}^{4}+%
\widehat{C}_{b3}^{4}\widehat{T}_{k4}^{3}+\widehat{C}_{b4}^{4}\widehat{T}%
_{k4}^{4}.
\end{eqnarray*}%
Using formulas (\ref{nontrdc}), (\ref{nontrtors}) and (\ref{aux3}), we
compute {\small
\begin{eqnarray*}
\ _{[1]}R_{3k} &=&\left( \widehat{L}_{3k}^{3}\right) ^{\ast }+\left(
\widehat{L}_{3k}^{4}\right) ^{\circ }=\left( \frac{\mathbf{\partial }%
_{k}h_{3}}{2h_{3}}-w_{k}\frac{h_{3}^{\ast }}{2h_{3}}\right) ^{\ast } = -w_{k}^{\ast }\frac{h_{3}^{\ast }}{2h_{3}}-w_{k}\left( \frac{h_{3}^{\ast }%
}{2h_{3}}\right) ^{\ast }+\frac{1}{2}\left( \frac{\mathbf{\partial }_{k}h_{3}%
}{h_{3}}\right) ^{\ast }, \\
\ _{[2]}R_{3k} &=&-\partial _{k}\widehat{C}_{3}+w_{k}\widehat{C}_{3}^{\ast
}+n_{k}\widehat{C}_{3}^{\circ }+\widehat{L}_{\ 3k}^{3\,}\widehat{C}_{3}+%
\widehat{L}_{\ 3k}^{4\,}\widehat{C}_{4}=-\partial _{k}(\frac{h_{3}^{\ast }}{%
2h_{3}}+\frac{h_{4}^{\ast }}{2h_{4}})+ \\
&&w_{k}(\frac{h_{3}^{\ast }}{2h_{3}}+\frac{h_{4}^{\ast }}{2h_{4}})^{\ast }+(%
\frac{\mathbf{\partial }_{k}h_{3}}{2h_{3}}-w_{k}\frac{h_{3}^{\ast }}{2h_{3}}%
)(\frac{h_{3}^{\ast }}{2h_{3}}+\frac{h_{4}^{\ast }}{2h_{4}})+\frac{1}{2}%
n_{k}^{\ast }\frac{\underline{h}_{4}^{\circ }}{2\underline{h}_{4}} \\
&=&w_{k}[\frac{h_{3}^{\ast \ast }}{2h_{3}}-\frac{3}{4}\frac{(h_{3}^{\ast
})^{2}}{(h_{3})^{2}}+\frac{h_{4}^{\ast \ast }}{2h_{4}}-\frac{1}{2}\frac{%
(h_{4}^{\ast })^{2}}{(h_{4})^{2}}-\frac{1}{4}\frac{h_{3}^{\ast }}{h_{3}}%
\frac{h_{4}^{\ast }}{h_{4}}]+n_{k}^{\ast }\frac{\underline{h}_{4}^{\circ }}{4%
\underline{h}_{4}} \\
&&+\left( \frac{\mathbf{\partial }_{k}h_{3}}{2h_{3}}+\frac{\mathbf{\partial }%
_{k}\underline{h}_{3}}{2\underline{h}_{3}}\right) (\frac{h_{3}^{\ast }}{%
2h_{3}}+\frac{h_{4}^{\ast }}{2h_{4}})-\frac{1}{2}\partial _{k}(\frac{%
h_{3}^{\ast }}{h_{3}}+\frac{h_{4}^{\ast }}{h_{4}}), \\
\ _{[3]}R_{3k} &=&\widehat{C}_{33}^{3}\widehat{T}_{k3}^{3}+\widehat{C}%
_{34}^{3}\widehat{T}_{k3}^{4}+\widehat{C}_{33}^{4}\widehat{T}_{k4}^{3}+%
\widehat{C}_{34}^{4}\widehat{T}_{k4}^{4} \\
&=&-\frac{h_{3}^{\ast }}{2h_{3}}\left[ \frac{\mathbf{\partial }_{k}h_{3}}{%
2h_{3}}-w_{k}\frac{h_{3}^{\ast }}{2h_{3}}-w_{k}^{\ast }{}\right] -\frac{%
h_{4}^{\ast }}{2h_{4}}\left[ \frac{\mathbf{\partial }_{k}(h_{4}\underline{h}%
_{4})}{2h_{4}\underline{h}_{4}}-w_{k}\frac{h_{4}^{\ast }}{2h_{4}}-n_{k}\frac{%
\underline{h}_{4}^{\circ }}{2\underline{h}_{4}}\right] \\
&=&w_{k}^{\ast }\frac{h_{3}^{\ast }}{2h_{3}}+w_{k}\left( \frac{(h_{3}^{\ast
})^{2}}{4(h_{3})^{2}}+\frac{(h_{4}^{\ast })^{2}}{4(h_{4})^{2}}\right) n_{k}%
\frac{h_{4}^{\ast }}{2h_{4}}\frac{\underline{h}_{4}^{\circ }}{2\underline{h}%
_{4}} -\frac{h_{3}^{\ast }}{2h_{3}}\frac{\mathbf{\partial }_{k}h_{3}}{2h_{3}}-%
\frac{h_{4}^{\ast }}{2h_{4}}\left( \frac{\mathbf{\partial }_{k}h_{4}}{2h_{4}}%
+\frac{\mathbf{\partial }_{k}\underline{h}_{4}}{2\underline{h}_{4}}\right).
\end{eqnarray*}%
} Summarizing, we get%
\begin{eqnarray}
\ \widehat{R}_{3k} &=&w_{k}\left[ \frac{h_{4}^{\ast \ast }}{2h_{4}}-\frac{1}{%
4}\frac{(h_{4}^{\ast })^{2}}{(h_{4})^{2}}-\frac{1}{4}\frac{h_{3}^{\ast }}{%
h_{3}}\frac{h_{4}^{\ast }}{h_{4}}\right] +n_{k}^{\ast }\frac{\underline{h}%
_{4}^{\circ }}{4\underline{h}_{4}}+n_{k}\frac{h_{4}^{\ast }}{2h_{4}}\frac{%
\underline{h}_{4}^{\circ }}{2\underline{h}_{4}}  \notag \\
&&+\frac{h_{4}^{\ast }}{2h_{4}}\frac{\mathbf{\partial }_{k}h_{3}}{2h_{3}}-%
\frac{1}{2}\frac{\partial _{k}h_{4}^{\ast }}{h_{4}}+\frac{1}{4}\frac{%
h_{4}^{\ast }\partial _{k}h_{4}}{(h_{4})^{2}}-\frac{h_{4}^{\ast }}{2h_{4}}%
\frac{\mathbf{\partial }_{k}\underline{h}_{4}}{2\underline{h}_{4}}.  \notag
\end{eqnarray}%
which is equivalent to (\ref{eq3}) if the conditions $n_{k}\underline{h}%
_{4}^{\circ }=\mathbf{\partial }_{k}\underline{h}_{4},$ see below formula (%
\ref{aux4}), are satisfied.

In a similar way, we compute $\ \widehat{R}_{4k}=\ _{[1]}R_{4k}+\
_{[2]}R_{4k}+\ _{[3]}R_{4k},$ where
\begin{eqnarray*}
&&\ _{[1]}R_{4k} =\left( \widehat{L}_{4k}^{3}\right) ^{\ast }+\left(
\widehat{L}_{4k}^{4}\right) ^{\circ },\ \ _{[2]}R_{4k} = -\partial _{k}%
\widehat{C}_{4}+w_{k}\widehat{C}_{4}^{\ast }+n_{k}\widehat{C}_{4}^{\circ }+%
\widehat{L}_{\ 4k}^{3\,}\widehat{C}_{3} \\
&& + \widehat{L}_{\ 4k}^{4\,}\widehat{C}_{4}, \ \ _{[3]}R_{4k} = \widehat{C}%
_{4d}^{a}\widehat{T}_{ka}^{d}=\widehat{C}_{43}^{3}\widehat{T}_{k3}^{3}+%
\widehat{C}_{44}^{3}\widehat{T}_{k3}^{4}+\widehat{C}_{43}^{4}\widehat{T}%
_{k4}^{3}+\widehat{C}_{44}^{4}\widehat{T}_{k4}^{4}.
\end{eqnarray*}
We get {\small
\begin{eqnarray*}
\ _{[1]}R_{4k}&=&-n_{k}^{\ast \ast }{}\frac{h_{4}}{2h_{3}}\underline{h}%
_{4}+n_{k}^{\ast }{}\left( -\frac{h_{4}^{\ast }}{2h_{3}}+\frac{h_{4}^{\ast
}h_{3}^{\ast }}{2(h_{3})^{2}}\right) \underline{h}_{4}+\frac{\mathbf{%
\partial }_{k}\underline{h}_{4}^{\circ }}{2\underline{h}_{4}}-\frac{%
\underline{h}_{4}^{\circ }\mathbf{\partial }_{k}\underline{h}_{4}}{2(%
\underline{h}_{4})^{2}}, \\
\ _{[2]}R_{4k} &=&w_{k}\left( -\frac{h_{4}^{\ast }}{2h_{4}}\frac{\underline{h%
}_{4}^{\circ }}{2\underline{h}_{4}}\right) -n_{k}^{\ast }{}\frac{h_{4}%
\underline{h}_{4}}{2h_{3}\underline{h}_{3}}\left( \frac{h_{3}^{\ast }}{2h_{3}%
}+\frac{h_{4}^{\ast }}{2h_{4}}\right) +n_{k}[\left( \frac{\underline{h}%
_{4}^{\circ }}{2\underline{h}_{4}}\right) ^{\circ }- \\
&&\frac{\underline{h}_{4}^{\circ }}{2\underline{h}_{4}}\frac{\underline{h}%
_{4}^{\circ }}{2\underline{h}_{4}}]+\frac{\mathbf{\partial }_{k}h_{4}}{2h_{4}%
}\frac{\underline{h}_{4}^{\circ }}{2\underline{h}_{4}}+\frac{\mathbf{%
\partial }_{k}\underline{h}_{4}}{2\underline{h}_{4}}\frac{\underline{h}%
_{4}^{\circ }}{2\underline{h}_{4}}-\frac{\partial _{k}\underline{h}%
_{4}^{\circ }}{2\underline{h}_{4}}+\frac{\underline{h}_{4}^{\circ }\partial
_{k}\underline{h}_{4}}{2(\underline{h}_{4})^{2}}, \\
\ _{[3]}R_{4k}&=&w_{k}\left( \frac{h_{4}^{\ast }}{2h_{4}}\frac{\underline{h}%
_{4}^{\circ }}{2\underline{h}_{4}}\right) +n_{k}(\frac{\underline{h}%
_{4}^{\circ }}{2\underline{h}_{4}})^{2}-\frac{\mathbf{\partial }_{k}h_{4}}{%
2h_{4}}\frac{\underline{h}_{4}^{\circ }}{2\underline{h}_{4}}-\frac{%
\underline{h}_{4}^{\circ }}{2\underline{h}_{4}}\frac{\mathbf{\partial }_{k}%
\underline{h}_{4}}{2\underline{h}_{4}}.
\end{eqnarray*}%
} We summarize the above three terms as follows {\small
\begin{eqnarray*}
&&\widehat{R}_{4k} =-n_{k}^{\ast \ast }{}\frac{h_{4}}{2h_{3}}\underline{h}%
_{4}+n_{k}^{\ast }{}\left( -\frac{h_{4}^{\ast }}{2h_{3}}+\frac{h_{4}^{\ast
}h_{3}^{\ast }}{2(h_{3})^{\ast }}-\frac{h_{4}^{\ast }h_{3}^{\ast }}{%
4(h_{3})^{\ast }}-\frac{h_{4}^{\ast }}{4h_{3}}\right) \underline{h}_{4} \\
&&+n_{k}\left( -\frac{\underline{h}_{4}^{\circ \circ }}{2\underline{h}_{4}}+%
\frac{(\underline{h}_{4}^{\circ })^{2}}{2(\underline{h}_{4})^{2}}+\left(
\frac{\underline{h}_{4}^{\circ }}{2\underline{h}_{4}}\right) ^{\circ }-\frac{%
\underline{h}_{4}^{\circ }}{2\underline{h}_{4}}\frac{\underline{h}%
_{4}^{\circ }}{2\underline{h}_{4}}+(\frac{\underline{h}_{4}^{\circ }}{2%
\underline{h}_{4}})^{2}\right) +\frac{\mathbf{\partial }_{k}\underline{h}%
_{4}^{\circ }}{2\underline{h}_{4}}-\frac{\underline{h}_{4}^{\circ }\mathbf{%
\partial }_{k}\underline{h}_{4}}{2(\underline{h}_{4})^{2}} \\
&& + \frac{\mathbf{\partial }_{k}h_{4}}{2h_{4}}\frac{\underline{h}%
_{4}^{\circ }}{2\underline{h}_{4}}+\frac{\mathbf{\partial }_{k}\underline{h}%
_{4}}{2\underline{h}_{4}}\frac{\underline{h}_{4}^{\circ }}{2\underline{h}_{4}%
}-\frac{\partial _{k}\underline{h}_{4}^{\circ }}{2\underline{h}_{4}}+\frac{%
\underline{h}_{4}^{\circ }\partial _{k}\underline{h}_{4}}{2(\underline{h}%
_{4})^{2}} -\frac{\mathbf{\partial }_{k}h_{4}}{2h_{4}}\frac{\underline{h}%
_{4}^{\circ }}{2\underline{h}_{4}}-\frac{\underline{h}_{4}^{\circ }}{2%
\underline{h}_{4}}\frac{\mathbf{\partial }_{k}\underline{h}_{4}}{2\underline{%
h}_{4}},
\end{eqnarray*}%
} and thus we have proved equations (\ref{eq4}).

For $\widehat{R}_{\ jka}^{i}=\frac{\partial \widehat{L}_{jk}^{i}}{\partial
y^{k}}-\left( \frac{\partial \widehat{C}_{ja}^{i}}{\partial x^{k}}+\widehat{L%
}_{lk}^{i}\widehat{C}_{ja}^{l}-\widehat{L}_{jk}^{l}\widehat{C}_{la}^{i}-%
\widehat{L}_{ak}^{c}\widehat{C}_{jc}^{i}\right) +\widehat{C}_{jb}^{i}%
\widehat{T}_{ka}^{b}$ from (\ref{dcurv}), we obtain zero values because $%
\widehat{C}_{jb}^{i}=0$ and $\widehat{L}_{jk}^{i}$ do not depend on $y^{k}.$
So, $\widehat{R}_{ja}=\widehat{R}_{\ jia}^{i}=0$.

Taking $\widehat{R}_{\ bcd}^{a}=\frac{\partial \widehat{C}_{.bc}^{a}}{%
\partial y^{d}}-\frac{\partial \widehat{C}_{.bd}^{a}}{\partial y^{c}}+%
\widehat{C}_{.bc}^{e}\widehat{C}_{.ed}^{a}-\widehat{C}_{.bd}^{e}\widehat{C}%
_{.ec}^{a}$ from (\ref{dcurv}) and contracting the indices in order to
obtain the Ricci coefficients, $\widehat{R}_{bc}=\frac{\partial \widehat{C}%
_{bc}^{d}}{\partial y^{d}}-\frac{\partial \widehat{C}_{bd}^{d}}{\partial
y^{c}}+\widehat{C}_{bc}^{e}\widehat{C}_{e}-\widehat{C}_{bd}^{e}\widehat{C}%
_{ec}^{d}$, we compute $\widehat{R}_{bc} =(\widehat{C}_{bc}^{3}) ^{\ast }+(%
\widehat{C}_{bc}^{4}) ^{\circ } -\partial _{c}\widehat{C}_{b}+\widehat{C}%
_{bc}^{3}\widehat{C}_{3}+\widehat{C}_{bc}^{4}\widehat{C}_{4} -\widehat{C}%
_{b3}^{3}\widehat{C}_{3c}^{3}-\widehat{C}_{b4}^{3}\widehat{C}_{3c}^{4}-%
\widehat{C}_{b3}^{4}\widehat{C}_{4c}^{3}-\widehat{C}_{b4}^{4}\widehat{C}%
_{4c}^{4}$. There are nontrivial values,{\small
\begin{eqnarray*}
\widehat{R}_{33} &=&\left( \widehat{C}_{33}^{3}\right) ^{\ast }+\left(
\widehat{C}_{33}^{4}\right) ^{\circ }-\widehat{C}_{3}^{\ast }+\widehat{C}%
_{33}^{3}\widehat{C}_{3}+\widehat{C}_{33}^{4}\widehat{C}_{4}-\widehat{C}%
_{33}^{3}\widehat{C}_{33}^{3}-2\widehat{C}_{34}^{3}\widehat{C}_{33}^{4} -%
\widehat{C}_{34}^{4}\widehat{C}_{43}^{4}  = \left( \frac{h_{3}^{\ast }}{2h_{3}}\right) ^{\ast }  \\ & & -\left( \frac{%
h_{3}^{\ast }}{2h_{3}}+\frac{h_{4}^{\ast }}{2h_{4}}\right) ^{\ast }+\frac{%
h_{3}^{\ast }}{2h_{3}}\left( \frac{h_{3}^{\ast }}{2h_{3}}+\frac{h_{4}^{\ast }%
}{2h_{4}}\right) -\left( \frac{h_{3}^{\ast }}{2h_{3}}\right) ^{2}-\left(
\frac{h_{4}^{\ast }}{2h_{4}}\right) ^{2} = -\frac{1}{2}\frac{h_{4}^{\ast \ast }}{h_{4}}+\frac{1}{4}\frac{%
(h_{4}^{\ast })^{2}}{(h_{4})^{2}}+\frac{1}{4}\frac{h_{3}^{\ast }}{h_{3}}%
\frac{h_{4}^{\ast }}{h_{4}}, \\
\widehat{R}_{44} &=&\left( \widehat{C}_{44}^{3}\right) ^{\ast }+\left(
\widehat{C}_{44}^{4}\right) ^{\circ }-\partial _{4}\widehat{C}_{4}+\widehat{C%
}_{44}^{3}\widehat{C}_{3}+\widehat{C}_{44}^{4}\widehat{C}_{4}-\widehat{C}%
_{43}^{3}\widehat{C}_{34}^{3}-2\widehat{C}_{44}^{3}\widehat{C}_{34}^{4} -%
\widehat{C}_{44}^{4}\widehat{C}_{44}^{4} \\
&=&-\frac{1}{2}\left( \frac{h_{4}^{\ast }}{h_{3}}\right) ^{\ast }\frac{%
\underline{h}_{4}}{\underline{h}_{3}}-\frac{h_{4}^{\ast }\underline{h}_{4}}{%
2h_{3}\underline{h}_{3}}\left( \frac{h_{3}^{\ast }}{2h_{3}}+\frac{%
h_{4}^{\ast }}{2h_{4}}-\frac{h_{4}^{\ast }}{h_{4}}\right) =-\frac{1}{2}\frac{h_{4}^{\ast \ast }}{h_{3}}\underline{h}_{4}+\frac{1}{4}%
\frac{h_{3}^{\ast }h_{4}^{\ast }}{(h_{3})^{2}}\underline{h}_{4}+\frac{1}{4}%
\frac{h_{4}^{\ast }}{h_{3}}\frac{h_{4}^{\ast }}{h_{4}}\underline{h}_{4}.
\end{eqnarray*}%
}

We get the nontrivial v--coefficients of the Ricci d--tensor,%
{\small
$$\widehat{R}_{~3}^{3} = \frac{1}{h_{3}\underline{h}_{3}}\widehat{R}_{33}=%
\frac{1}{2h_{3}h_{4}}[-h_{4}^{\ast \ast }+\frac{(h_{4}^{\ast })^{2}}{2h_{4}}+%
\frac{h_{3}^{\ast }h_{4}^{\ast }}{2h_{3}}]\frac{1}{\underline{h}_{3}},\
\widehat{R}_{~4}^{4} = \frac{1}{h_{4}\underline{h}_{4}}\widehat{R}_{44}=%
\frac{1}{2h_{3}h_{4}}[-h_{4}^{\ast \ast }+\frac{(h_{4}^{\ast })^{2}}{2h_{4}}+%
\frac{h_{3}^{\ast }h_{4}^{\ast }}{2h_{3}}]\frac{1}{\underline{h}_{3}},$$
}
i.e., equation (\ref{eq2}).

\subsection{Zero torsion conditions}

Let us analyze how to solve the equation
\begin{equation*}
\widehat{T}_{4k}^{4}=\widehat{L}_{4k}^{4}-e_{4}(N_{k}^{4})=\frac{\mathbf{%
\partial }_{k}(h_{4}\underline{h}_{4})}{2h_{4}\underline{h}_{4}}-w_{k}\frac{%
h_{4}^{\ast }}{2h_{4}}-n_{k}\frac{\underline{h}_{4}^{\circ }}{2\underline{h}%
_{4}}=0,
\end{equation*}%
which follows from formulas (\ref{nontrtors}) for a vanishing torsion for \ $%
\widehat{\mathbf{D}}.$ Taking any $\underline{h}_{4}$ for which
\begin{equation}
n_{k}\underline{h}_{4}^{\circ }=\mathbf{\partial }_{k}\underline{h}_{4},
\label{aux4}
\end{equation}%
 the condition $n_{k}\frac{h_{4}^{\ast }}{2h_{4}}\frac{\underline{h}%
_{4}^{\circ }}{2\underline{h}_{4}}-\frac{h_{4}^{\ast }}{2h_{4}}\frac{\mathbf{%
\partial }_{k}\underline{h}_{4}}{2\underline{h}_{4}}=0$ is satisfied. For
instance, parameterizing $\underline{h}_{4}=~^{h}\underline{h}_{4}(x^{k})%
\underline{h}(y^{4}),$ the equation (\ref{aux4}) is solved by any
 $\underline{h}(y^{4})=e^{\varkappa y^{4}}$ and $n_{k}=\varkappa \partial
_{k}[~^{h}\underline{h}_{4}(x^{k})]$,  for $\varkappa =cont$.

We conclude that for any $n_{k}$ and $\underline{h}_{4}$ related by
conditions (\ref{aux4}) the zero torsion conditions (\ref{nontrtors}) are
the same as for $\underline{h}_{4}=const.$ Using a similar proof from \cite%
{vexsol1,vexsol2}, it is possible to verify by straightforward computations
that $\widehat{T}_{\beta \gamma }^{\alpha }=0$ if the equations (\ref{lccond}%
) are solved.

\end{document}